\documentclass[journal,twoside]{IEEEtran}

\usepackage{amsmath,amsfonts,amscd,amssymb,amsthm,bm,bbm,latexsym,mathdots,mathtools,fge}
\usepackage{algorithm,algorithmicx,algpseudocode}
\usepackage{microtype}
\usepackage{hyperref}
\usepackage[nameinlink]{cleveref}
\usepackage{ulem}
\usepackage{scalerel,stackengine} 
\usepackage{relsize} 
\usepackage{comment}
\usepackage{graphicx, wrapfig, tikz, pgfplots}

\numberwithin{equation}{section}
\allowdisplaybreaks
\hypersetup{
    colorlinks=true,%
    citecolor=blue,%
    filecolor=blue,%
    linkcolor=blue,%
    urlcolor=blue
}

\newcommand*{\myfont}{\fontfamily{phv}\selectfont} 

\newtheorem{theorem}{Theorem}[section]
\newtheorem{lemma}[theorem]{Lemma}

\newtheorem{proposition}[theorem]{Proposition}

\newtheorem{remark}[theorem]{Remark}

\renewenvironment{proof}{\noindent {\bf Proof.} }{\endprf\par}
\def \endprf{\hfill {\vrule height6pt width6pt depth0pt}\medskip}

\renewcommand{\mathbf}{\boldsymbol}
\newcommand{\mb}{\mathbf}
\newcommand{\mr}{\mathrm}
\newcommand{\mc}{\mathcal}

\newcommand{\bb}{\mathbb}

\newcommand{\R}{\bb R}

\newcommand{\set}[1]{\left\{ #1 \right\}}

\newcommand{\Brac}[1]{\left\lbrace #1 \right\rbrace}
\newcommand{\brac}[1]{\left[ #1 \right]}
\newcommand{\paren}[1]{ \left( #1 \right) }

\newcommand{\wh}{\widehat}
\newcommand{\wt}{\widetilde}

\stackMath
\newcommand\widecheck[1]{%
\savestack{\tmpbox}{\stretchto{%
  \scaleto{%
    \scalerel*[\widthof{\ensuremath{#1}}]{\kern-.6pt\bigwedge\kern-.6pt}%
    {\rule[-\textheight/2]{1ex}{\textheight}}
  }{\textheight}%
}{0.5ex}}%
\stackon[1pt]{#1}{\scalebox{-1}{\tmpbox}}%
}
\newcommand{\wc}{\widecheck}

\def\-{\raisebox{.75pt}{-}} 

\DeclareMathOperator{\sign}{sign}

\newcommand{\norm}[2]{\left\| #1 \right\|_{#2}}
\newcommand{\abs}[1]{\left| #1 \right|}
\newcommand{\innerprod}[2]{\left\langle #1,  #2 \right\rangle}
\newcommand{\prob}[1]{\bb P\left[ #1 \right]}

\newcommand{\eps}{\varepsilon}
\newcommand{\E}{\bb E}

\newcommand{\simiid}{\sim_{\mr{i.i.d.}}}

\newcommand{\1}{\mathbf 1}

\newcommand*\acr[1]{\textscale{.85}{#1}}
\newcommand*\bull{\raisebox{-0.365em}[-1em][-1em]{\textscale{4}{$\cdot$}} }

\usetikzlibrary{svg.path}
\definecolor{orcidlogocol}{HTML}{A6CE39}
\tikzset{
  orcidlogo/.pic={
    \fill[orcidlogocol] svg{M256,128c0,70.7-57.3,128-128,128C57.3,256,0,198.7,0,128C0,57.3,57.3,0,128,0C198.7,0,256,57.3,256,128z};
    \fill[white] svg{M86.3,186.2H70.9V79.1h15.4v48.4V186.2z}
                 svg{M108.9,79.1h41.6c39.6,0,57,28.3,57,53.6c0,27.5-21.5,53.6-56.8,53.6h-41.8V79.1z M124.3,172.4h24.5c34.9,0,42.9-26.5,42.9-39.7c0-21.5-13.7-39.7-43.7-39.7h-23.7V172.4z}
                 svg{M88.7,56.8c0,5.5-4.5,10.1-10.1,10.1c-5.6,0-10.1-4.6-10.1-10.1c0-5.6,4.5-10.1,10.1-10.1C84.2,46.7,88.7,51.3,88.7,56.8z};
  }
}

\newcommand\orcidicon[1]{\href{https://orcid.org/#1}{\mbox{\scalerel*{
\begin{tikzpicture}[yscale=-1,transform shape]
\pic{orcidlogo};
\end{tikzpicture}
}{|}}}}

\begin{document}

\title{Compressed Sensing Microscopy with Scanning Line Probes}
\author{\IEEEauthorblockN{Han-Wen Kuo\textsuperscript{\orcidicon{0000-0003-2989-5218}}, 
Anna E. Dorfi\textsuperscript{\orcidicon{0000-0002-8045-3311}},
Daniel V. Esposito\textsuperscript{\orcidicon{0000-0002-0550-801X}} and  
John N. Wright\textsuperscript{\orcidicon{0000-0003-2683-4428}}, \IEEEmembership{Member, IEEE} }  

\thanks{Manuscript received September 25, 2019. The authors would like to acknowledge funding support from the Columbia University SEAS Interdisciplinary Research Seed (SIRS) Funding program, and NSF CDS\&E 1710400. (\emph{Corresponding author: Han-Wen Kuo}). }%

\thanks{H.-W. Kuo and J. N. Wright are with Department
of Electrical Engineering and Data Science Institute, Columbia University, New York, NY, 10025. Email: hk2673@columbia.edu and jw2966@columbia.edu.}%
\thanks{A. E. Dorfi and D. V. Esposito are with Department of Chemical Engineering, Electrochemical Energy Center and Lenfest Center for Sustainable Energy, Columbia University, New York, NY, 10025. Email: de2300@columbia.edu and aed2521@columbia.edu.}
}

\markboth{}{Kuo \MakeLowercase{\textit{et al.}}: Compressed Sensing Microscopy with Scanning Line Probes}

\maketitle

\begin{abstract} 
In applications of scanning probe microscopy, images are acquired by raster scanning a point probe across a sample. Viewed from the perspective of compressed sensing (CS), this pointwise sampling scheme is inefficient, especially when the target image is structured. While replacing point measurements with delocalized, incoherent measurements has the potential to yield order-of-magnitude improvements in scan time, implementing the delocalized measurements of CS theory is challenging. In this paper we study a partially delocalized probe construction, in which the point probe is replaced with a continuous line, creating a sensor which essentially acquires line integrals of the target image. We show through simulations, rudimentary theoretical analysis, and experiments, that these line measurements can image sparse samples far more efficiently than traditional point measurements, provided the local features in the sample are enough separated. Despite this promise, practical reconstruction from line measurements poses additional difficulties: the measurements are partially coherent, and real measurements exhibit nonidealities. We show how to overcome these limitations using natural strategies (reweighting to cope with coherence, blind calibration for nonidealities), culminating in an end-to-end demonstration. 
\end{abstract}


 

\section{Introduction}

\IEEEPARstart{S}{canning} probe microscopy (SPM) is a fundamental technique for imaging interactions between a probe and the sample of interest.  Unlike traditional optical microscopy, the resolution achievable by SPM is not constrained by the diffraction limit, making SPM especially advantageous for nanoscale, or atomic level imaging, which has widespread applications in chemistry, biology and materials science \cite{wiesendanger1994scanning}. Conventional implementations of SPM typically adopt a raster scanning strategy, which utilizes a probe with small and sharp tip, to form a pixelated heatmap image via point-by-point measurements from interactions between the probe tip and the surface. Despite its capability of nanoscale imaging, SPM with point measurements is inherently slow, especially when scanning a large area or producing high-resolution images.

When the target signal is highly structured, compressed sensing (CS) \cite{donoho2006compressed, candes2008introduction, foucart2017mathematical} suggests it is possible to design a data acquisition scheme in which the number of measurements is largely dependent on the signal complexity, instead of the signal size, from which the signal can be efficiently reconstructed algorithmically. In nanoscale microscopy, images are often spatially sparse and structured. {CS} theory suggests for such signals, localized measurements such as pointwise samples are inefficient. In contrast, delocalized, spatially spread measurements are better suited for reconstructing a sparse image.

\begin{figure}[t!]
	\centering
\begin{tikzpicture}
	\node[rectangle, draw=black, inner sep=0.7pt, line width = 0.7pt] at (0in, 0in){\includegraphics[width=0.22\textwidth,height =0.17\textwidth]{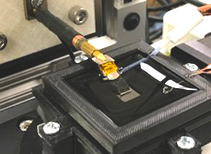}};
	\node[rectangle, draw=black, inner sep=0.7pt, line width = 0.7pt] at (0.25\textwidth,0in) {\includegraphics[width=0.22\textwidth,height =0.17\textwidth]{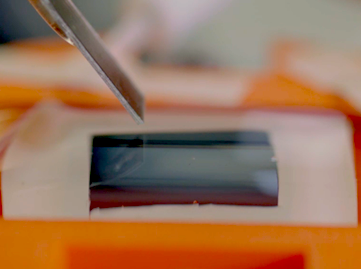}};
\end{tikzpicture}
	
	\caption{Scanning electrochemical microscope with continuous line probe. Left: the lab made {SECM} device with line probe, mounted on an automated probe arms with a rotating sample stage. Right: closeup side view of the line probe near the sample surface. }\label{fig:microscope}  
\end{figure}

Unfortunately, the dense (delocalized) sensing schemes suggested by {CS} theory (and used in other applications, e.g., \cite{lustig2008compressed,studer2012compressive,veeraraghavan2011coded}) are challenging to implement in the settings of micro/nanoscale imaging. Motivated by these concerns, \cite{o2018scanning} introduced a new type of {\em semilocalized} probe, known as a {\em line probe}, which integrates the signal intensity along a straight line, and studied it in the context of a particular microscopy modality known as scanning electrochemical microscopy (SECM) \cite{bard1980electrochemical, bard1991chemical}. In {SECM} with \emph{line probe}, the working end of the probe consists of a straight line, which produces a single measurement by collecting accumulated redox reaction current induced by the probe and sample. These line measurements are semilocalized, sample a spatially sparse image more efficiently than measurements from point probes, and ``have an edge'' for high resolution imaging  since a thin and sharp line probe can be manufactured with ease. Moreover, experiments in \cite{o2018scanning} suggest that a combination of line probes and compressed sensing reconstruction could potentially yield order-of-magnitude reductions in imaging time for sparse samples.

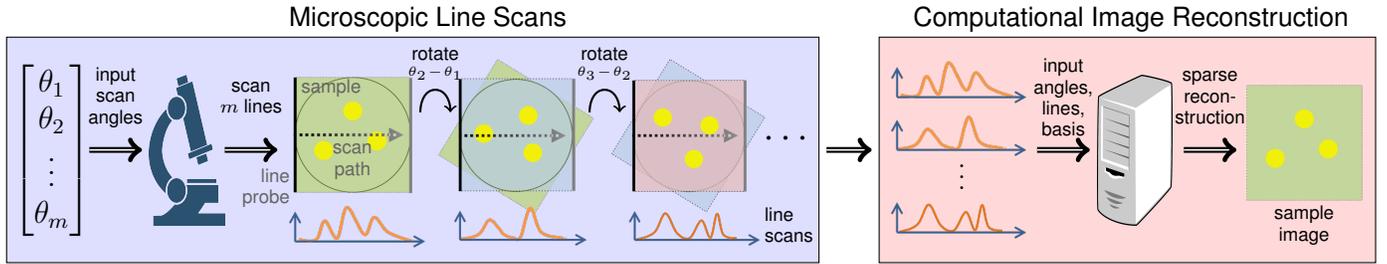
\begin{figure*}[htp!]  
	\centering 
	\begin{tikzpicture}  
		\sffamily{
		\fill[blue!13!white,draw=black] (0,0) rectangle (10.8,3);     
		\node at (0.6,1.5) {\large$\begin{bmatrix}
			\theta_1 \\ \theta_2  \\ \vdots \\ \theta_m  
		\end{bmatrix}$};      
		\draw[line width=0.8pt,->,double] (1.1,1.5) -- (1.8,1.5);
		\node[align=center] at (1.45,2.2) { \scriptsize input  \\[-0.14cm]  \scriptsize scan \\[-0.14cm]   \scriptsize angles};   
		\node at (2.4,1.5) {\includegraphics[width=1cm, height=2cm]{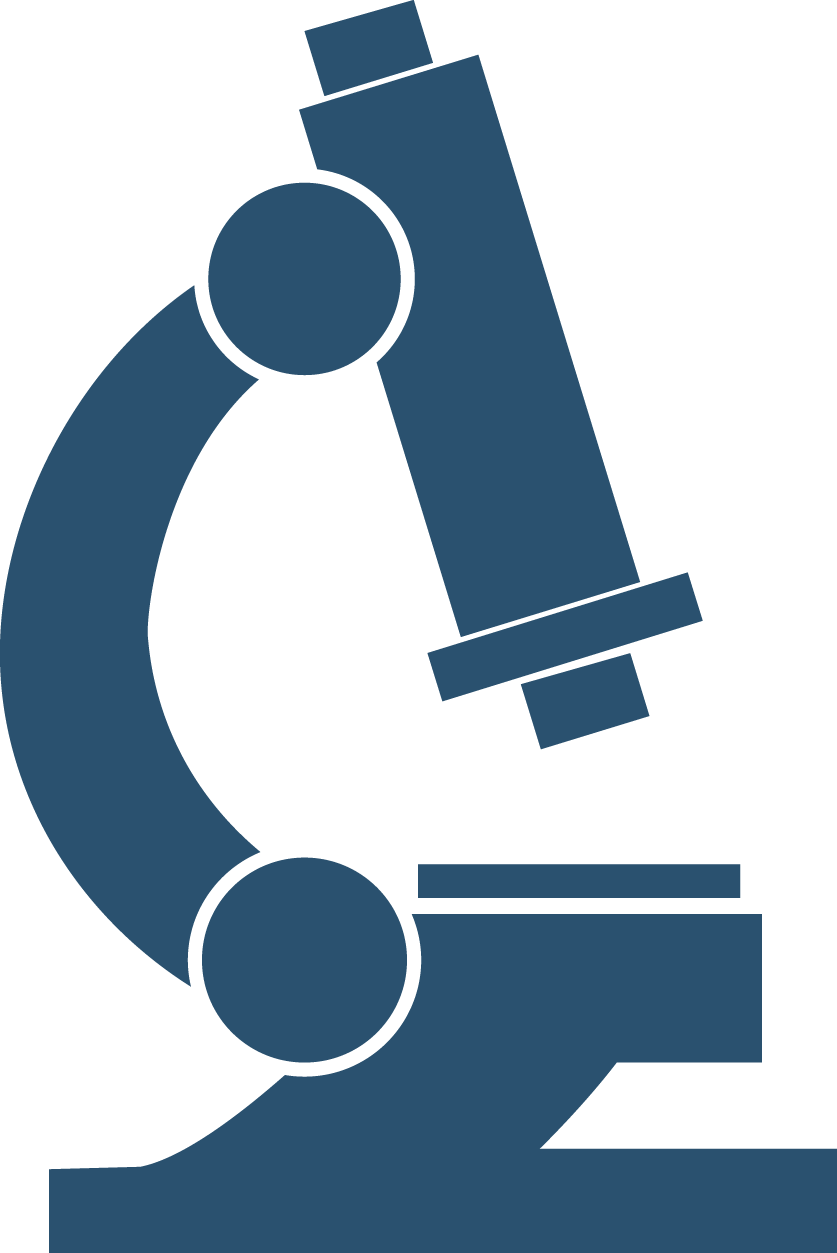}};
		\draw[line width=0.8pt,->,double] (2.9,1.5) -- (3.6,1.5); 
		\node[align=center] at (3.25,2.2) { \scriptsize  {scan}  \\[-0.14cm]  \scriptsize {$m$ lines}};  
		\node at (4.7,1.3)  {\includegraphics[width=2cm]{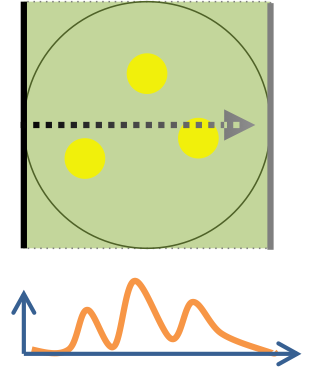}};
		\node at (4.35,2.32){\scriptsize {\color[RGB]{100,100,100} \myfont{sample} } };  
		\node[align=right] at (3.45,1) { \scriptsize {\color[RGB]{100,100,100} \myfont{line}} \\[-0.14cm] \scriptsize {\color[RGB]{100,100,100}{probe}}};
		\node[align=center] at (4.6,1.35) { \scriptsize {\color[RGB]{100,100,100} {scan}} \\[-0.14cm] \scriptsize  {\color[RGB]{100,100,100} {path}}}; 
		\draw[line width = 0.8pt,->] (5.5,2.0) to[out=90, in=180] (5.7,2.3) to[out=0, in=90](5.9,2.0); 
		\node[align=center] at (5.7,2.65) {\scriptsize {rotate} \\[-0.2cm] \tiny$\theta_2\!-\!\theta_1$}; 
		\node at (6.8,1.43) {\includegraphics[width=2.1cm]{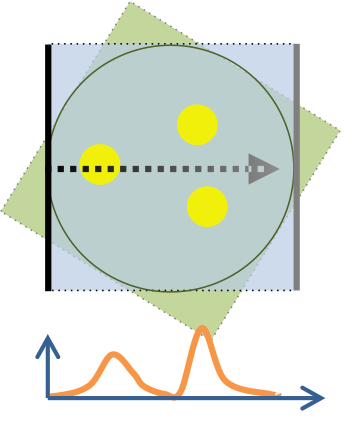}};
		\draw[line width = 0.8pt,->] (7.75,2.0) to[out=90, in=180] (7.95,2.3) to[out=0, in=90](8.15,2.0); 
		\node[align=center] at (7.95,2.65) {\scriptsize {rotate} \\[-0.2cm] \tiny$\theta_3\!-\!\theta_2$}; 
		\node at (9.1,1.43)  {\includegraphics[width=2.1cm]{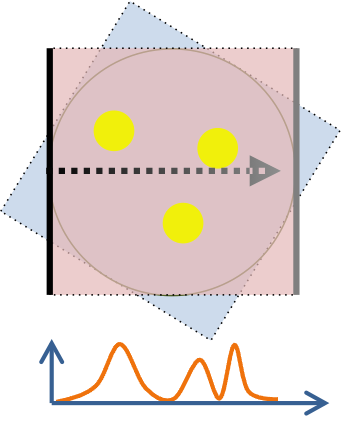}};
		\node at (10.4,1.65) {\Large$\cdots$};
		\node[align =left] at (10.4,0.5) {\scriptsize  {line} \\[-0.14cm] \scriptsize  {scans}};  
		\draw[line width=0.8pt,->,double] (10.9,1.5) -- (11.5,1.5);
		\fill[red!15!white,draw=black]  (11.6,0) rectangle (18.2,3);
		\node at (12.7,2.4) {\includegraphics[width=2cm]{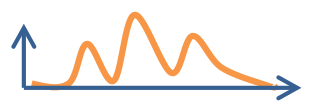}};
		\node at (12.7,1.7) {\includegraphics[width=2cm]{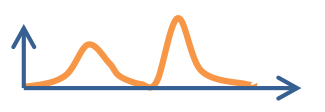}}; 
		\node at (12.7,0.6) {\includegraphics[width=2cm]{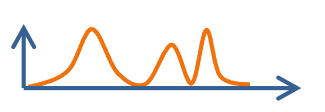}};
		\node at (12.7,1.25) {{$\vdots$}};
		\draw[line width=0.8pt,->,double] (13.7,1.5) -- (14.4,1.5);
		\node[align=center] at (14.05,2.2) {\scriptsize  {input} \\[-0.14cm] \scriptsize {angles,}\\[-0.14cm] \scriptsize {lines,}\\[-0.14cm] \scriptsize {basis} }; 
		\node at (15.0,1.5) {\includegraphics[width=1cm, height=2cm]{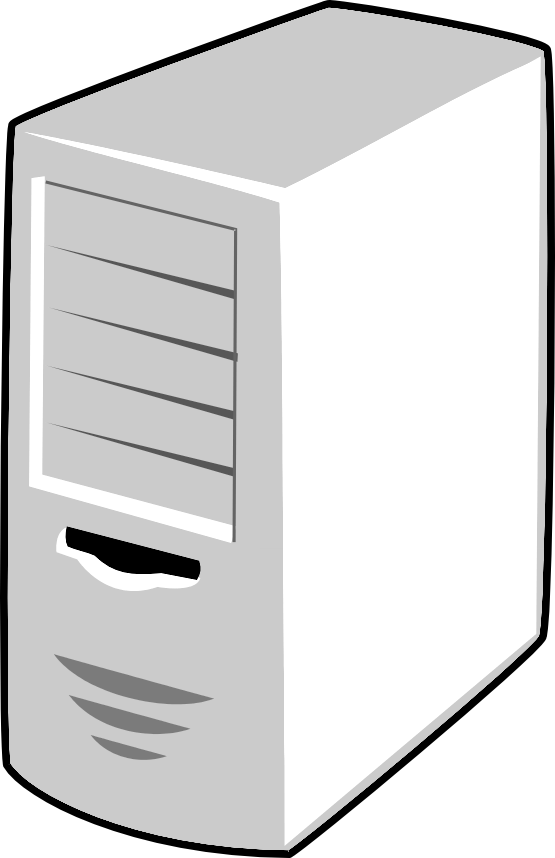}};
		\draw[line width=0.8pt,->,double] (15.65,1.5) -- (16.35,1.5);
		\node[align=center] at (16,2.2){\scriptsize  {sparse}  \\[-0.14cm] \scriptsize {recon-} \\[-0.14cm] \scriptsize {struction}}; 
		\node at (17.25,1.6) {\includegraphics[width=1.53cm]{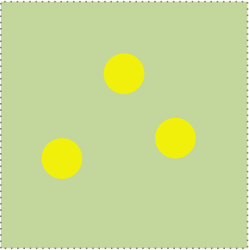}};
		\node[align=center] at (17.25, 0.5){\scriptsize {sample} \\[-0.14cm] \scriptsize {image}};
		\node at (5.6,3.24)  {{Microscopic Line Scans}};
		\node at (14.9,3.24) {{
		Computational Image Reconstruction}};
		}
	\end{tikzpicture} 
	\caption{ Scanning procedure of SECM with continuous line electrode probe. The user begins with mounting the sample on a rotational stage of microscope and chooses $m$ scanning angles. The microscope then carries on sweeping the line probe across the sample,  and measures the accumulated current generated between the interreaction of probe and the sample at equispaced  intervals of moving distance. After a sweep ends, the sample is rotated to another scanning angle and the scanning sweep procedure repeats, until all $m$ line scans are finished. Collecting all scan lines, and providing the information of the scanning angles, the microscope system parameters (such as the point spread function) and the sparse representing basis of image, the final sample image is produced via computation with sparse reconstruction algorithm.} \label{fig:line_scan_procedure} 
\end{figure*} 

 Realizing the promise of line probes (both in SECM and in microscopy in general) demands a more careful study of the mathematical and algorithmic problems of image reconstruction from line scans. Because these measurements are structured, they deviate significantly from conventional CS theory, and basic questions such as the number of line scans required for accurate reconstruction are currently unanswered. Moreover, practical reconstruction from line scans requires modifications to accommodate nonidealities in the sensing system. In this paper, we will address both of these questions through rudimentary analysis and experiments, showing that if the local features are either small or separated, then stable image reconstruction from line scans is attainable.

In the following, we will first describe the scanning procedure and introduce a mathematical model for line scans  in \Cref{sec:line-scans}.  \Cref{sec:line_scans_properties} discusses several important properties of this measurement model, including a rudimentary study of compressed sensing reconstruction with line scans of a spatially sparse image. In \Cref{sec:algorithm}, we give a practical algorithm for reconstruction from line scans, which accommodates measurement nonidealities. Finally, \Cref{sec:real-data} demonstrates our algorithm and theory by efficiently reconstructing both simulated and real SECM examples.

\subsection{Contribution}  

\begin{itemize}
	\item We expose the lowpass property of line scans, and with rudimentary analysis showing that the exact reconstruction of a sparse image is possible with only three line scans provided these features are well-separated.  
	\item We describe the challenges associated with image reconstruction from practical line scans, due to the high coherence of measurement model and inaccurate estimate of point-spread-function. Our reconstruction algorithm addresses these issues.
	\item Based on this theory and algorithmic ideas, we demonstrate a complete algorithm for image reconstruction of SECM with line scans, which includes an efficient algorithm for computation of image reconstruction, yielding improved results compared to \cite{o2018scanning}.
\end{itemize}

\subsection{Related work}   

%
%

\subsubsection{Compressed sensing tomography}

Line measurements also arise in \emph{computational tomography} (CT) imaging, a classical imaging modality (see e.g., \cite{hounsfield1973computerized, kak1979computerized, herman2009fundamentals}), with great variety of applications ranging from medical imaging to material science \cite{wellington1987x,frank1992electron,duric2005development}. Classical CT reconstruction recovers an image from densely sampled line scans, by approximately solving an inverse problem \cite{nuyts1998iterative, shepp1974fourier}.  These methods do not incorporate the prior knowledge of the structure of the target image, and degrade sharply when only a few CT scans are available. Compressed sensing offers an attractive means of reducing the number of measurements needed for accurate CT image reconstruction, and has been employed in applications ranging from medical imaging to (cryogenic) electron transmission microscopy \cite{chen2008prior,malczewski2013pet, goris2012atomic,  saghi2011three, leary2013compressed, donati2017compressed, binev2012compressed, nicoletti2013three}. The dominant approach assumes that the target image is sparse in a Fourier or wavelet basis, and reconstructs it via $\ell^1$ minimization or related techniques. Images in SECM and related modalities typically exhibit much stronger structure: they often consist some number of small particles \cite{davis2010evidence,batista2015nonadditivity}, or other repeated motifs \cite{cheung2018dictionary}. In this situation, CS is especially promising. On the other hand, as we will see below, understanding the interaction between line scans and spatially localized features demands that we move beyond conventional CS theory.

\subsubsection{Mathematical theory of line scans: Radon transform and image super-resolution} 

The question of recoverability from line measurements is related to the theory of the {\em Radon transform}, which corresponds to a limiting situation in which line scans at every angle are available \cite{radon20051,cormack1963representation, natterer2001tomography}. The Radon transform is invertible, meaning perfect reconstruction is possible (albeit not stable) in this limiting situation. Due to the \emph{projection slice theorem} \cite{helgason2010integral}, the line projections are inherently lowpass, and so the line scan reconstruction problem is related to superresolution imaging  \cite{farsiu2004advances}. When the image of interest consists of sparse point sources, the image can be stably recovered from its low-frequency components, provided the point sources are sufficiently separated  \cite{candes2014towards}. Similarly, we can hope to achieve stable recovery of localized features from line scans as long as the features are sufficiently separated.

\section{Line scan measurement model}\label{sec:line-scans}

To implement line scans for {SECM}, a line probe (\Cref{fig:microscope}) is mounted on an automated arm which positions the probe onto the sample surface. The line scan signal is generated by placing this line probe in different places, and measuring the integrated current induced by the interaction between the line probe and the electroactive part of the sample.  In a pragmatic scanning procedure (\Cref{fig:line_scan_procedure}), the user will choose distinct scanning angles $\theta_1, \dots, \theta_m$. For each angle $\theta$, the line probe is oriented in direction $\mb u_{\theta} = (\cos\theta,\sin\theta)$ and swept along the normal direction $\mb u_\theta^\perp = (\sin\theta,-\cos\theta)$. Each sweep of probe generates the projection of the target image along the probe direction $\mb u_{\theta}$; collecting these projections for each $\theta_i$, we obtain our complete set of measurements.

\subsection{Line projection}
To describe the scanning procedure more precisely, we begin with a mathematical idealization, in which the probe measures a line integral of the image. In this model, when the probe body is oriented in direction $\mb u_\theta$ at position $t$, we observe the integral of the image over $ \ell_{\theta,t}:= \{\mb w\in\R^2 \,\big|\, \langle\mb u_\theta^\perp,\,\mb w\rangle = t\} $:
\begin{align}\label{eqn:line_project}
	\mc L_\theta[\mb Y](t) &\,:= \,\textstyle\int_{\ell_{\theta,t}}\mb Y(\mb w)\, d{\mb w} \notag \\
	&\,=\,\textstyle\int_{s} \mb Y\paren{s\cdot\mb u_{\theta} \,+\, t\cdot \mb u_{\theta}^\perp}\, ds.
\end{align} 
Collecting these measurements for all $t$, we obtain a function $\mc L_{\theta}[\mb Y]$ which is the projection of the image along the direction $\mb u_\theta$. We refer to the operation $\mc L_\theta : L^2(\R^2) \to L^2 (\R)$ as a {\em line projection}. Combining projections in $m$ directions $\Theta = \set{\theta_i}_{i=1}^m$, we obtain an operator $\mc L_\Theta : L^2(\R^2)\to L^2(\R \times [m])$:
\begin{align}\label{eqn:line_project}
	\mc L_\Theta[\mb Y] &\,:=\,\tfrac{1}{\sqrt m}\brac{\,\mc L_{\theta_1}[\mb Y],\ldots,\mc L_{\theta_m}[\mb Y]\,}.
\end{align}

\subsection{Line scans}
In reality, it is not possible to fabricate an infinitely sharp line probe, and hence our measurements do not correspond to ideal line projections. The line probe has a response in its normal direction, causing a blurring effect that can be modeled as convolution with point spread function $\mb \psi$ along the sweeping direction. In SECM, $\mb \psi$ is typically skewed with a long tail in the sweeping direction. Accounting for this effect is important for obtaining accurate reconstructions in practice. In this more realistic model, our measurements $\wt{\mb R} \in L^2(\R\times[m])$ become
\begin{align}\label{eqn:line_scans}
	\wt{\mb R} &\,=\, \tfrac{1}{\sqrt m}[\mb \psi*\mc L_{\theta_1}\brac{\mb Y},\ldots,\mb \psi*\mc L_{\theta_m}\brac{\mb Y}] \notag \\
	&\,=:\, \mb \psi * \mc L_\Theta\brac{\mb Y}.    
\end{align}  
This measurement consists of $m$ functions $\mb \psi * \mc L_{\theta_i} \brac{\mb Y}(t)$ of a single (real) variable $t$, which corresponds to the translation of the probe in the $\mb u_{\theta_i}^\perp$ direction. In practice, we do not measure this function at every $t$, but rather collect $n$ equispaced samples. Write the sampling operator as $\mc S:L^2[\R] \to \R^n$, then our discretized  measurements $\mb R_i$ with scanning angle $\theta_i$ is defined as $\mb R_i = \mc S\{\wt{\mb R}_i\}$. Collect all $m$ discrete line scans, the final measurement $\mb R\in\R^{n\times m}$ is written as
\begin{align}
	\mb R &\,=\, [\mc S\{\wt{\mb R}_1\},\ldots, \mc S\{\wt{\mb R}_m\}] \,=:\, \mc S\{\wt{\mb R}\}.
\end{align}
Our task is to understand when and how we can reconstruct the target image $\mb Y$ from these samples.

\begin{figure}[t!]  
\centering	 
\begin{tikzpicture} 
	\fill[blue!13!white, draw=black!30!white] (-1.3in,-0.99in) rectangle (-1.9in+\columnwidth, 1.14in);
	\node at (0in,0in){\includegraphics[width = 0.28\textwidth]{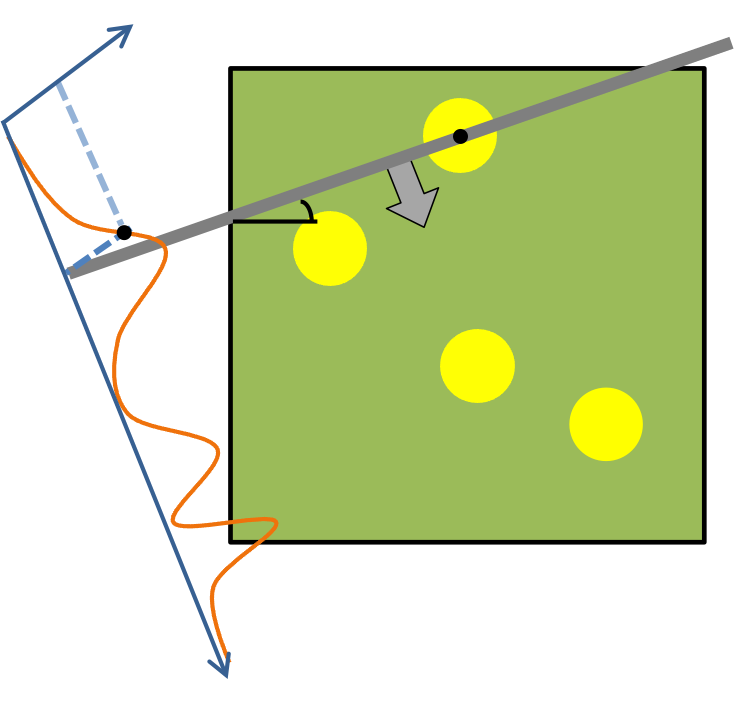} };
	\node at (-0.5in,0.95in) {$\mb R_\theta(t)$}; 
	\node at (-0.56in,-0.72in) {$t$};
	\node at (-0.17in,0.27in) {$\theta$};
	\node at (0.41in,0.26in) {$(\sin\theta,-\cos\theta)$};
	\node at (0.88in,0.91in) {$\ell_{\theta,t}$};
	\node at (0.33in, 0.48in) {$\mb w_i$}; 
	\node at (0.3in,-0.63in) {$\mb Y$}; 
\end{tikzpicture}
\caption{Mathematical expression of a single measurement from the line probe. When the stage rotate by $\theta$ clockwise, the relative rotation of probe to sample is counterclockwise by $\theta$. The grey line in the figure represents the rotated line probe, orienting in direction $\mb u_\theta = (\cos\theta,\sin\theta)$, and is sweeping in direction $\mb u_\theta^\perp = (\sin\theta,-\cos\theta)$. When it comes across the point $\mb w_i$ where $t=\langle\mb u_\theta^\perp,\,\mb w_i\rangle$, it integrates over the contact region $\ell_{\theta,t }$ between the probe and substate and produces a measurement $\mb  R_\theta(t)$.}\label{fig:single_line_project}
\end{figure}
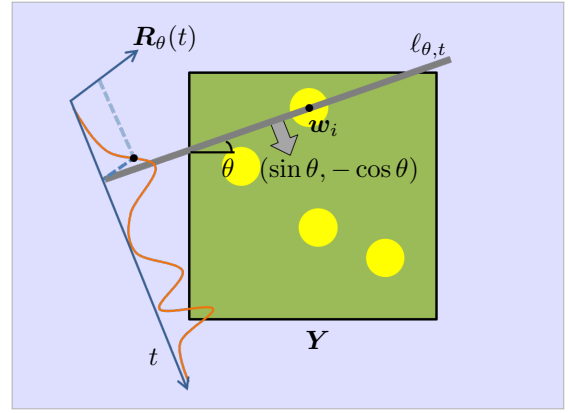


\section{Promises and problems of line scans}\label{sec:line_scans_properties}

The line measurements $\mc L_\theta$  enjoy two major advantages as an imaging model: (i) compared to pointwise measurements, the line projections are more delocalized, hence  can be more efficient while measuring a spatially sparse signal; and (ii) it is easier to build a sharp edge for the line probe (even sharper then the tip diameter of a point probe), which is well-suited to detect ultra-high frequency components in the probe sweeping direction. This makes possible fast and high resolution imaging for scanning microscopes. 
 
Nevertheless, the line projection comes with a few apparent disadvantages. Consider a limiting scenario, in which infinitely many line projections are available, corresponding to every angle in $[0,2\pi)$. The {\em projection slice theorem} implies that these measurements are invertible, and the image can be perfectly reconstructed. However, this reconstruction is {\em not stable}: viewed in Fourier domain, these measurements are approximately lowpass, and inverting them amplifies high frequencies. This that means even though a single line projection can be highly sensitive to directional high frequency components, the cumulative line projections are not. Stably inverting them requires prior knowledge of the image to be reconstructed. Moreover, to reconstruct an image consisting of multiple localized features, these features need to be either sufficiently separated.

The other disadvantage of line projections can be viewed from the CS perspective, that the line scans measurements model are not coherent---\emph{even} if the local features are well separated. This means that in practice, when using only a few line scans for reconstruction, the number of lines required for exact reconstruction cannot be obtained from conventional CS theory. More importantly, the coherence of line projections can cause issues in image reconstruction; conventional methods tend to produce reconstructions with incorrect magnitudes. 
 
Finally, we discuss measurement nonidealities due to variability in the PSF $\mb \psi$. In the next section, we will provide an algorithmic solution addressing both issues from the coherence of line projections and incomplete information of PSF.

\subsection{CS of line projections for highly localized image} 

Compressed sensing, in its simplest form,  asserts that if the target signal has sparse representation, it can be exactly reconstructed from a few measurements, provided those measurements are incoherent to the basis of sparsity. Since in microscopic imaging the underlying signal is often structured and spatially localized, CS theory suggests that delocalized measurements, such as line projections, could yield more efficient reconstructions than conventional point measurements. 

Inspired by CS, we start from providing the sufficient conditions of sparse image reconstruction from line measurements via total variation minimization \cite{krahmer2017total}. Later, base on these conditions, we demonstrate the the use of line probes can be indeed more efficient than using point probes.  

 \begin{proposition} \label{prop:dualcert_condition}\emph{[Certificate of TV-norm minimization]} Let $\mb X_0 = \sum_{\mb w\in \mc W}\alpha_{\mb w}\mb\delta_{\mb w} $\footnote{The Dirac measure $\mb\delta$ satisfies   $\int \mb D(\mb w)\mb \delta_{\mb w_i}(d\mb w) = \mb D(\mb w_i)$ for continuous and compactly supported $\mb D$ and has total variation $\int \abs{\mb\delta_{\mb w_i}}(d\mb w)  = 1$, so $\mb D*\mb \delta_{\mb w}$ represents $\mb D$ with center at $\mb w$ \cite{rudin2006real}. As a functional, we write $\innerprod{\mb \delta_{\mb w_i}}{\cdot} : L^2(\R^2) \to\R$ where $\innerprod{\mb\delta_{\mb w_i}}{\mb D} = \mb D(\mb w_i)$.} with $\abs{\mc W}<\infty$. Given continuous compactly supported circular symmetric $\mb D\in L^2(\R^2)$, scanning angles $\Theta = \set{\theta_1,\ldots,\theta_m}$ and measurement $\wt{\mb R} = \mc L_\Theta[\mb D*\mb X_0]$.   Suppose there exists $\wt{\mb Q}$ as finite sum of weighed Diracs  such that 
	\begin{align}\label{eqn:dual_condition}
	\begin{cases}
		\mb D* \mc L_{\Theta}^*\big[\wt{\mb Q}\big](\mb w) = \sign\paren{\alpha_{\mb w}}, & \quad \mb w\in\mc W\\
	\abs{\mb D* \mc L_{\Theta}^*\big[\wt{\mb Q}\big](\mb w)} < 1, &\quad \mb w\not\in\mc W.
	\end{cases}
	\end{align} 
	If the Gram matrix $\mb G\in\R^{\abs{\mc W}\times \abs{\mc W}}$, defined as 
	\begin{align}\label{eqn:gram_G}
		\mb G_{ij} = \innerprod{\mc L_\Theta[\mb D*\mb\delta_{\mb w_i}]}{\,\mc L_\Theta[\mb D*\mb\delta_{\mb w_j}]},\; \mb w_i,\mb w_j\in\mc W
	\end{align}
	is positive definite, then $\mb X_0$ is the unique optimal solution to 
	\begin{align}\label{eqn:tv-min}
		\textstyle\min_{\mb X\in \mr{BV}(\R^2)} \textstyle\int_{\mb w}\abs{\mb X}(d\mb w) \quad s.t.\quad \wt{\mb R} = \mc L_\Theta[\mb D*\mb X].
	\end{align}
\end{proposition} 
 
 \begin{proof}
	First we show the existence result. Note that $\mb X_0$ satisfies the equalitiy constraint \eqref{eqn:tv-min}  automatically, and since total variation of Dirac measure is exactly one,
	\begin{align}
		\textstyle\int{\abs{\mb X_0}(d\mb w')} &\;=\; \textstyle\sum_{\mb w}\textstyle \int\abs{\alpha_{\mb w}} \mb \delta_{\mb w}(d\mb w') \;=\; \textstyle \sum_{\mb w} \abs{\alpha_{\mb w}} \notag \\
		&\;=\;   \textstyle \sum_{\mb w}\alpha_{\mb w}\cdot \mb D*\mc L_{\Theta}^*\big[\wt{\mb Q}\big](\mb w)\notag 
		\end{align}
then since $\mb D$ is circular symmetric, $\mb D*\mc L^*_\Theta\big[\wt{\mb Q}\big](\mb w) = \langle\mb \delta_{\mb w},\mb D*\mc L^*_\Theta\big[\wt{\mb Q}\big]\rangle = \langle\mc L_\Theta\brac{\mb D*\mb \delta_{\mb w}},\wt{\mb Q}\rangle$, we derive 
\begin{align}
		\textstyle\int{\abs{\mb X_0}(d\mb w')} &\;=\; \langle\mc L_\Theta\brac{\mb D*\textstyle\sum_{\mb w}\alpha_{\mb w}\mb \delta_{\mb w}},\wt{\mb Q}\rangle \;=\; \langle\wt{\mb R},\wt{\mb Q}\rangle \label{eqn:dualgap0} \notag 
	\end{align}
which certifies that $\mb X_0$ is an optimal solution to the problem since the duality gap $\int\abs{\mb X_0}(d\mb w')-\langle\wt{\mb R},\wt{\mb Q}\rangle = 0$. For uniqueness, let $\mb X'= \sum_{\mb w'\in \mc W'}\alpha'_{\mb w'}\mb\delta_{\mb w'}$ to be another optimal solution with $\mc W'\not\subseteq\mc W$, since we know $\mb X'$ is primal feasible $\wt{\mb R} = \mc L_\Theta\brac{\mb D*\mb X'}$, then 
\begin{align}
	\textstyle \int\abs{\mb X_0}(d\mb w') &\;=\; \langle\wt{\mb R},\wt{\mb Q}\rangle \;=\;  \langle\mc L_\Theta\brac{\mb D*\mb X'},\wt{\mb Q}\rangle \notag \\
	&\;=\; \langle\mb X',\mb D*\mc L_\Theta^*[\wt{\mb Q}]\rangle \notag \\
	&\;=\; \textstyle\sum_{\mb w'\in\mc W'}\alpha'_{\mb w'}\mb D*\mc L^*_\Theta\big[\wt{\mb Q}\big](\mb w') \notag
	\end{align}
and by knowing $\mc W'\not\subseteq\mc W$ and using the second condition in \eqref{eqn:dual_condition}:  
\begin{align} 
	\textstyle\int\abs{\mb X_0}(d\mb w') &\;<\;  \textstyle\sum_{\mb w'\in\mc W'}\abs{\alpha_{\mb w'}'} \;=\; \textstyle\int\abs{\mb X'}(d\mb w') \notag 
\end{align}
thus $\mb X'$ is an optimal solution only if  $\mc W'\subseteq \mc W$. Finally   uniqueness of $\mb X_0$ follows from  injectivity of $\mc L_{\Theta}[\mb D*\,\cdot\;]$ over $\mc W$ from \eqref{eqn:gram_G}. 
\end{proof}

Specifically, when the target image is highly spatially sparse and its components are well separated, the line projections can be a very efficient measurement model. A concrete example is demonstrated in \Cref{thm:small_separated_disc}, where we assume the sparse component of the image signal are small and separated discs; if the radius of the discs are sufficiently small, then,  perhaps surprisingly, only three line projections is required to exactly reconstruct the image via efficient algorithm.

\begin{lemma}\emph{[Reconstruction from three line projection]}  \label{thm:small_separated_disc} Consider an image consists of $ k\geq 2$ discs radius $r$. If the centers $\mb w_1,\ldots\mb w_k$ are at least separated by $\tfrac{2}{C}k^2r$, then three continuous line projections with probe direction chosen independent uniformly at random suffice to recover the image with probability at least $1-C$ via solving \eqref{eqn:tv-min}. 
\end{lemma}
\begin{proof}
We first argue that with high probability, no pair of discs overlaps within any line scan. Let $\theta_i\simiid\mr{Unif}[-\pi,\pi)$ denote the $i$-th scanning angle. Write $d$ as the minimum distance between all pairs of $(\mb w_i,\mb w_j)$,  the probability that any particular pair of two discs overlap is bounded as
\begin{align}
	&\prob{\text{Two discs overlap on line scan}   \,\wt{\mb R}_i} \notag \\
	&\quad\leq\, \prob{\theta_i\in \brac{-\sin^{-1}\paren{\tfrac{2r}{d}}, \sin^{-1}\paren{\tfrac{2r}d}}}\notag \\
	&\quad =\, \tfrac2\pi \sin^{-1}\tfrac{2 r}d
\end{align}                        
Using the assumption that $R<\frac{d}8$ to bound $\sin^{-1}(\frac{2r}d) < \frac{2\pi r}{3d}$ and summing the failure probability over all three line scans and $\frac{k(k-1)}2$ pairs of discs, we obtain:
\begin{align}\label{eqn:small_disc_no_overlap}
	&\prob{\text{Two of the $k$ discs overlap at either} \, \wt{\mb R}_1,\wt{\mb R}_2,\wt{\mb R}_3} \notag \\
	&\quad\leq\,  \tfrac{3k^2}2\cdot\prob{\text{Two discs overlap on line scan}   \,\wt{\mb R}_1} \notag \\                                             
&\quad\leq\, \tfrac{3k^2}{\pi}  \sin^{-1}\paren{\tfrac{2r}d}\,\leq\, \tfrac{2k^2 r}d \notag \\
&\quad\leq\, C
\end{align}                                                                          
Thus, with probability at least $1-C$, no pair of discs overlaps in any line scan.

Since there are no overlapping discs in any line, a single line projection $\wt{\mb R}_i(t)$  with scan angle $\theta_i$ has largest magnitude at points $t$ where the probe body passes the disc center  $\mb w_j$. These points of largest magnitude $\beta_j$ is located at $\langle\mb u^\perp_{\theta_i}, \mb w_j\rangle$ on $\wt{\mb R}_i$, or equivalently,
\begin{align}\label{eqn:dualcert_position}
	\mc L_{\theta_i}[\mb D*\mb\delta_{\mb w_j}](\langle\mb u_{\theta_i}^\perp,\mb w_j\rangle) = \beta_j,\qquad i=1,2,3
\end{align}
Using these points, we construct the dual certificate $\wt{\mb Q}_i$ for angle $\theta_i$, where 
\begin{align}
	\wt{\mb Q}_i = \textstyle\sum_{j=1}^k \tfrac{1}{\sqrt3\beta_j}\mb\delta_{\langle\mb u_{\theta_i}^\perp,\mb w_j\rangle}
\end{align}
and $\wt{\mb Q} = \big[\wt{\mb Q}_1,\wt{\mb Q}_2,\wt{\mb Q}_3\big]$. Using this certificate we verify \eqref{eqn:dual_condition} holds. For the equality, calculate at every $\mb w_j\in\set{\mb w_1,\ldots\mb w_k}$:
\begin{align}
	&\,\mb D*\mc L^*_{\set{\theta_1,\theta_2,\theta_3}}\big[\wt{\mb Q}\big](\mb w_j) \notag \\
	=&\, \langle\mb D*\mc L^*_{\set{\theta_1,\theta_2,\theta_3}}\big[\wt{\mb Q}\big],\mb \delta_{\mb w_j}\rangle\,=\,\tfrac1{\sqrt3} \textstyle\sum_{i=1}^3 \langle\wt{\mb Q}_i,\mc L_{\theta_i}[\mb D*\mb \delta_{\mb w_j}]\rangle  \notag \\ 
	=&\, \tfrac{1}{\sqrt3}\textstyle\sum_{i=1}^3\big\langle\tfrac{1}{\sqrt3\beta_j}\mb\delta_{\langle\mb u_{\theta_i}^\perp,\mb w_j\rangle},\mc L_{\theta_i}[\mb D*\mb \delta_{\mb w_j}]\big\rangle \notag \\
	=&\, \tfrac{1}{3\beta_j}\textstyle\sum_{i=1}^3\mc L_{\theta_i}[\mb D*\mb\delta_{\mb w_j}](\langle\mb u_{\theta_i}^\perp,\mb w_j\rangle) = 1 \label{eqn:3l-equality}
\end{align}
where the third line is by plugging in $\wt{\mb Q}$ and derived with no overlap property; the forth line via plugging in \eqref{eqn:dualcert_position}. For the inequality,  calculate 
\begin{align}
	&\,\abs{\mb D*\mc L_{\set{\theta_1,\theta_2,\theta_3}}^*\big[\wh{\mb Q}\big](\mb w)} \notag \\
	&\qquad\;\;=\,\abs{\textstyle\sum_{i=1}^3\textstyle\sum_{j=1}^k\tfrac{1}{3\beta_j}\mc L_{\theta_i}[\mb D*\mb\delta_{\mb w}](\langle\mb u_{\theta_i}^\perp,\mb w_j\rangle)} \label{eqn:3l_ineq_1}
\end{align}
which is derived similarly as \eqref{eqn:3l-equality}.  Now, by observing $\mc L_\theta[\mb D*\mb\delta_{\mb w}]$ has unique local maximum at $\langle\mb u_{\theta}^\perp,\mb w\rangle$, each summand (w.r.t. $i$) in \eqref{eqn:3l_ineq_1} is strictly less than 1 if $\mb w$ does \emph{not} satisfy  
\begin{align}
	\exists\,j\in\set{1,\ldots,k},\;\;\; \langle\mb u_{\theta_i}^\perp,\,\mb w\rangle = \langle\mb u_{\theta_i}^\perp,\,\mb w_j\rangle.
\end{align}
Accordingly, define the \emph{back projection} line set $\ell_{\theta_i,t_j}$ as
\begin{align}\label{eqn:line_set}
 \ell_{\theta_i, t_j }:= \{\mb w\in\R^2 \,\big|\, \langle\mb u_{\theta_i}^\perp,\,\mb w\rangle = \langle\mb u_{\theta_i}^\perp,\,\mb w_j\rangle\},
\end{align} 
we want to show that for every $\mb w\not\in\set{\mb w_1,\ldots,\mb w_k}$, 
\begin{align} \label{eqn:3l_ineq_2}
	\mb w\,\not\in\, \cap_{i=1}^3\big(\cup_{j=1}^k\ell_{\theta_i,t_j}\big)
\end{align} 
then \eqref{eqn:3l_ineq_1} is strictly less then 1.

W.l.o.g., write $\mb w_{j\ell} = \ell_{\theta_1,t_j}\cap \ell_{\theta_2,t_\ell}$.  Suppose the point $\mb w_{j\ell}$ is in the third line set $\cup_{j=1}^k \ell_{\theta_3,t_j}$, then there exists some disc center $\mb w_q$ such that $\langle\mb u_{\theta_3}^\perp,\,\mb w_{j\ell}\rangle = \langle\mb u_{\theta_3}^\perp,\,\mb w_q\rangle$. Since $\theta_3$ is generated uniform randomly, we conclude that for any $j,\ell$:
\begin{align}
	\prob{\,\exists\, q \in 1,\ldots,k\quad s.t.\quad\mb w_{j\ell} \in\ell_{\theta_3,t_q}\,}=0.  
\end{align}                  
The direction $\mb u_{\theta_3}$ is not aligned with the line formed by $\mb w_{j\ell},\mb w_q$ almost surely. This proves \eqref{eqn:3l_ineq_2}.

Finally, the diagonal entries of Gram matrix $\mb G$ defined in \eqref{eqn:gram_G} is strictly positive, and the off-diagonal entries $\mb G_{ij}$ can be derived as
\begin{align}
	\mb G_{ij}&\,=\,\textstyle\tfrac13\sum_{t=1}^3\innerprod{\mc L_{\theta_t}[\mb D]*\mb\delta_{\mb w_i} }{\mc L_{\theta_t}[\mb D]*\mb\delta_{\mb w_j}} \,=\,0
\end{align}
by no overlap property. Hence $\mb G$ is positive definite. This concludes that solving total variation minimization successfully reconstruct the image from three line projections.
\end{proof}

The proof idea can be depicted pictorially in \Cref{fig:proof_three_lines}, in which we show the construction of dual certificate $\wt{\mb Q}$, and the  \emph{back projection} operation $\mc L_\Theta^*$ on $\wt{\mb Q}$ which we used in the proof to certify the optimality. In fact, as we will show later, the operation $\mc L_\Theta^*$ is  the cornerstone for most of the reconstruction algorithms in computed tomography, as well as in our sparse reconstruction algorithm.

\begin{figure}[t!]
	\centering 
	\begin{tikzpicture}
		\fill[blue!13!white, draw=black!30!white] (-1.3in,-1.3in) rectangle (-1.9in+\columnwidth, 1.12in); 
		\node at (0.1in,0in) {\includegraphics[width= 0.31\textwidth]{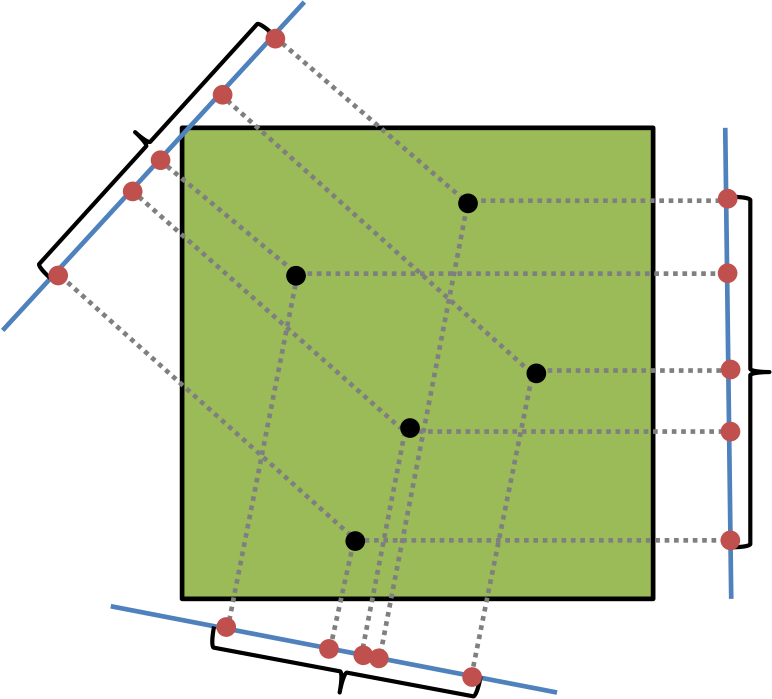}};
		\node[rotate=-12] at (0in,-1.1in) {$\cup_{j=1}^k \ell_{\theta_1,t_j}$};
		\node at (-0.6in,-0.9in) {$\wt{\mb R}_1$};
		\node[rotate=50] at (-0.73in,0.67in) {$\cup_{j=1}^k \ell_{\theta_2,t_j}$};
		\node at (-0.06in,0.89in) {$\wt{\mb R}_2$};
		\node[rotate=270] at (1.35in,-0.1in) {$\cup_{j=1}^k \ell_{\theta_3,t_j}$}; 
		\node at (1.22in,-0.8in) {$\wt{\mb R}_3$};
		\node at (0.53in,0.32in) {$\mb w_1$};
		\node at (-0.12in,-0.51in) {$\mb w_k$};
	\end{tikzpicture} 
	\caption{Proof sketch for sufficiency of image recovery from three line projections. Given a sample with separated tiny discs $\mb w_1,\ldots \mb w_k$ (black dots), randomly choosing three lines projection forms lines $\wt{\mb R}_1,\wt{\mb R}_2,\wt{\mb R}_3$, in which all the discs after line projection (red dots) are well-separated. From each of these lines, we construct the dual $\wt{\mb Q}$ as center of red dots, and a back projection image form the dual (dash lines), forming the set $\cup_{j=1}^k \ell_{\theta_i,t_j}.$ Intersection of three such line sets is exactly the set of ground truth disc centers.}\label{fig:proof_three_lines}
\end{figure}
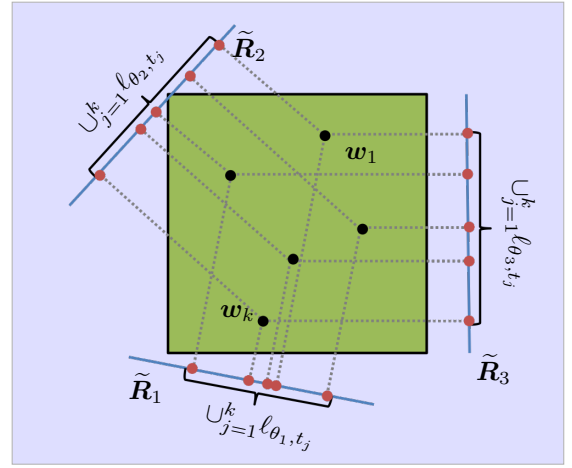  

\subsection{Reconstructability from line projections of localized image in practice} \label{sec:line-projection-low-pass}
While the  microscopic images are often sparse in the spatial domain, they rarely satisfy the conditions of  \Cref{thm:small_separated_disc}, in which the local features are uncharacteristically small and far apart. In the following, we will show in practical application of line scans, when the image consists of multiple localized motifs, the performance of line measurements degrades once the ratio between the size of motifs and its separating distance increases.

\subsubsection{Coherence of line projection of two localized motif}\label{sec:high-coherence} We start from a simple case considering an image with two motifs located at different locations. Define a $2\times 2$ Gram matrix $\mb G$ with its $ij$-th entries being  \emph{coherence} \cite{donoho2006stable} between line projected signal of two motifs $\mb D$ with center at $\mb w_i$ and $\mb w_j$ respectively, 
\begin{align} \label{eqn:Gij}
	\mb G_{ij} = \innerprod{\mc L_\Theta[\mb D*\mb\delta_{\mb w_i}]}{\,\mc L_\Theta[\mb D*\mb\delta_{\mb w_j}]}.
\end{align}
If the off-diagonal entry $\mb G_{ij}$ is small in magnitude compared to the diagonal entries $\mb G_{ii},\mb G_{jj}$, then it suffices to reconstruct the image exactly with efficient algorithm. Conversely, if $\mb G$ is ill-conditioned or even rank-deficient, then exact recovery will be impossible.


\begin{lemma}\emph{[Coherence of line projection Gaussians]}\label{lem:coherence_gauss} Let $\mb D$ be the two-dimensional Gaussian functions with covariance $r\mb I^2$ and normalized in a sense that $\norm{\mc L_0[\mb D]}{L^2} = 1$. If $\theta$ is     uniformly random, then the expectation of inner product between two line projected $\mb D$ at different locations $\mb w_i,\mb w_j$ is bounded by  
\begin{align}
	&\paren{1-\tfrac{d^2}{8r^2}}\1_{\set{d\leq 2r}} + \tfrac{r}{2d}\1_{\set{d>2r}}\notag \\
	&\; \leq \E_\theta \big\langle\mc L_\theta[\mb D*\mb\delta_{\mb w_i}],\,\mc L_\theta[\mb D*\mb\delta_{\mb w_j}]\big\rangle \leq \tfrac{1}{\sqrt{1+d^2/4r^2}}. 
\end{align} 
where $d = \norm{\mb w_i-\mb w_j}2$. 
\end{lemma}

\begin{proof}
	Write $\mb d(t) = \mc L_0\big[\mb D\big](t)$, where $\mb d$ is a one-dimensional standard Gaussian with deviation $r$. Since $\mb D$ is circular symmetric, the line projection of $\mb D$ in any angle is identical, that is, $\mc L_\theta[\mb D] = \mc L_0[\mb D]$ for every $\theta$. Also write $\mb w_i-\mb w_j = d(\cos\phi,\sin\phi)$, then 
\begin{align}\label{eqn:injective_oned_gauss}
&\big\langle\mc L_\theta[\mb D*\mb\delta_{\mb w_i}],\,\mc L_\theta[\mb D*\mb\delta_{\mb w_j}]\big\rangle \notag \\
	&\qquad\qquad  =\, \big\langle\mc L_\theta[\mb D] * \mc L_\theta[\mb\delta_{\mb w_i}],\, \mc L_\theta[\mb D]*\mc L_\theta[\mb\delta_{\mb w_j}] \big\rangle \notag\\
	&\qquad\qquad =\,  \Big\langle\,\mb d*\mb d ,\,\mb \delta_{\abs{\mb u_\theta^*(\mb w_i-\mb w_j)}}\Big\rangle\notag\\  
	&\qquad\qquad =\,  \paren{\mb d*\mb d}\paren{d\cos(\theta-\phi)} \notag \\
	&\qquad\qquad =\,\exp\paren{\tfrac{-d^2\cos^2(\theta-\phi)}{4r^2}},
\end{align}   
where the first equality is by interchanging iterated integrals; the second equality is by knowing the adjoint of convolution is correlation and $\mb d$ is symmetric; and the final equality is by observing that $\mb d*\mb d$ is a Gaussian function with variance $\sqrt 2 r$ and $(\mb d*\mb d)(0) = 1$. 

We derive the  expectation upper bound of \eqref{eqn:injective_oned_gauss} over $\theta$ as 
\begin{align} 
	&\E_{\theta} \innerprod{\mc L_\theta[\mb D*\mb\delta_{\mb w_i}]}{\,\mc L_\theta[\mb D*\mb\delta_{\mb w_j}]}  \notag \\
	&\qquad\qquad =\,  \tfrac{1}{\pi}\textstyle\int_{0}^{\pi} \exp\paren{-d^2\cos^2\theta)/4r^2}\,d\theta \notag \\
	&\qquad\qquad \leq\, \tfrac{1}{\pi}\textstyle\int_{0}^\pi\tfrac{1}{1+(d^2\cos^2\theta)/4r^2}\, d\theta \notag \\
	&\qquad\qquad = \,\tfrac{1}{\sqrt{1+d^2/4r^2}} \label{eqn:injective_exp_theta}
\end{align}
by utilizing $
	\exp(-x^2)(1+x^2)<1$ in the second inequality. 
	
As for the lower bound of \eqref{eqn:injective_oned_gauss}, from its first equality, we calculate when $d > 2 r$, then
\begin{align}
	&\tfrac{1}{\pi}\textstyle\int_{0}^{\pi} \exp\big(-(d^2\cos^2\theta)/4r^2\big)\,d\theta \notag \\
	&\quad\quad \geq\, \tfrac1\pi \textstyle\int_{0}^{\pi}\exp\paren{-(d^2/4r^2)\cdot (2r^2/d^2)}\cdot\1_{\set{\cos^2\theta\leq 2r^2/d^2}}\,d\theta\notag \\
	&\quad\quad \geq\, \tfrac2\pi\cdot \exp\paren{-\tfrac12}\cdot\Big(\tfrac\pi2 - \cos^{-1}\sqrt{2r^2/d^2} \Big) \notag \\
	&\quad\quad \geq\, \tfrac2\pi\cdot\exp\paren{-\tfrac12}\cdot  (\sqrt2r/d)  \notag \\ 
	&\quad\quad \geq\,  r / 2d. 
\end{align}
using $\cos^{-1}x\leq \frac\pi 2 - x$ for $x\in[0, 0.5]$. And when $d \leq  2 r$, we simply have 
\begin{align}
	&\tfrac{1}{\pi}\textstyle\int_{0}^{\pi} \exp\big(-(d^2\cos^2\theta)/4r^2\big)\,d\theta \geq 1-d^2/8r^2 
\end{align}
via Taylor expansion at $d/2r = 0$
\end{proof}


\Cref{lem:coherence_gauss} shows the coherence between line projections of two bell-shaped motif with radius $\approx\! r$ and center distance $d$ is dominated by the distance-to-diameter ratio $d/2r$. Because of the projection slice theorem, the matrix $\E_\theta\mb G$ is always positive definitive. However, its condition number greatly increases when the image  consists of highly overlapping local features. When the ratio is small, say $d/2r < 1$, in which the two projected motifs are overlapping, then  $\E_\theta\mb G_{ij}$ will be close to one as with the diagonals, implies $\E_\theta\mb G$ become severely ill-conditioned even in the two-sparse case. Generally speaking, line scans are  not CS-theoretical optimal sampling method for recovering images consisting of superposing discs.

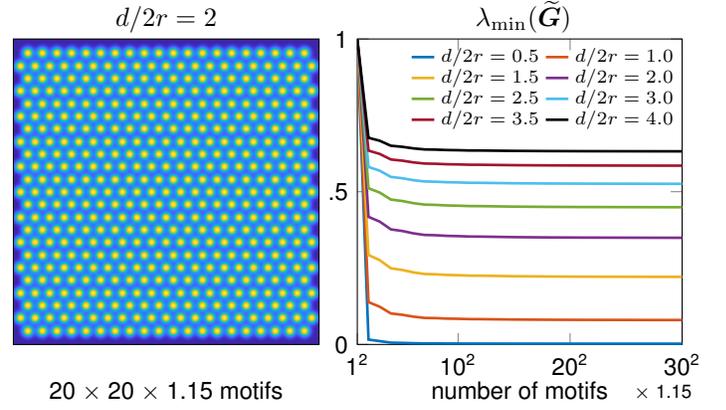
\begin{figure}
\centering
%
%
\definecolor{mycolor1}{rgb}{0.00000,0.44700,0.74100}%
\definecolor{mycolor2}{rgb}{0.85000,0.32500,0.09800}%
\definecolor{mycolor3}{rgb}{0.92900,0.69400,0.12500}%
\definecolor{mycolor4}{rgb}{0.49400,0.18400,0.55600}%
\definecolor{mycolor5}{rgb}{0.46600,0.67400,0.18800}%
\definecolor{mycolor6}{rgb}{0.30100,0.74500,0.93300}%
\definecolor{mycolor7}{rgb}{0.63500,0.07800,0.18400}%
\definecolor{mycolor8}{rgb}
	{0,0,0}%

\newcommand{\plotwidth}{1.7in}
\newcommand{\plotheight}{1.6in}
\newcommand{\plotxone}{0}
\newcommand{\plotxtwo}{1.8in}
\newcommand{\plotyone}{0}

\begin{tikzpicture} 

\sffamily{

\begin{axis}[%
width=\plotheight,
height=\plotheight,
at={(\plotxone,\plotyone)},
scale only axis,
point meta min=0.0367283487797067,
point meta max=0.986995949431156,
axis on top,
xmin=-1.005,
xmax=1.005, 
xtick={\empty}, 
ytick={\empty},
xlabel style={yshift=0.06in},
xlabel={\small{20 $\times$ 20 $\times$ 1.15 motifs}},
ymin=-1.005,
ymax=1.005,
title style={yshift=-0.1in},
title={\myfont{$d/2r = 2$}},
axis background/.style={fill=white},
legend style={legend cell align=left, align=left, draw=white!15!black},
]
\addplot [forget plot] graphics [xmin=-1.005, xmax=1.005, ymin=-1.005, ymax=1.005] {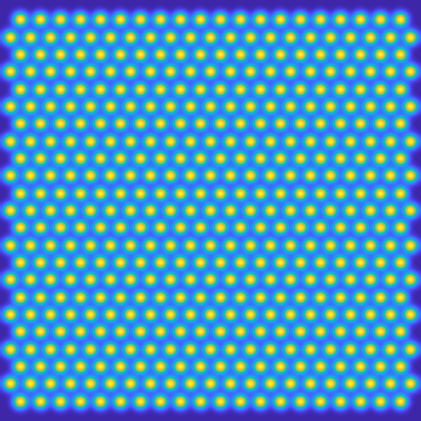};
\end{axis}

 
\begin{axis}[%
width=\plotwidth,
height=\plotheight,
at={(\plotxtwo,\plotyone)},
scale only axis,
xmin=1,
xmax=30,
ymin=0,
ymax=1,
ytick={0,.5,1}, 
xtick={1,10,20,30},
xticklabels={1\textsuperscript{2} ,10\textsuperscript{2},20\textsuperscript{2},30\textsuperscript{2}},
yticklabels={0,.5,1},
x tick label style={font=\small},
y tick label style={font=\small},
axis background/.style={fill=white},
xlabel style={yshift=0.06in},
xlabel = {\small{number of motifs}},
title style={xshift = 0.02in, yshift=-0.1in},
title={\myfont{$\lambda_{\mr{min}}(\wt{\mb G})$}},
legend image post style={scale=0.5},
legend style={at={(\plotwidth-0.72in,\plotheight+0.01in)},legend cell align=left, anchor=north, legend columns=2, draw=none, fill=none, font=\scriptsize, row sep=-0.1cm},  
]
\addplot [color=mycolor1, line width=1.0pt]
  table[row sep=crcr]{%
1	1\\
2	0.0158967288505295\\
3	0.0109273034655624\\
4	0.00580136850205513\\
5	0.00496604152984505\\
6	0.00396170067339666\\
7	0.00332868591432761\\
8	0.00324192629402198\\
9	0.0029982414542931\\
10	0.00293019615431881\\
11	0.00281617299614331\\
12	0.00277678373520093\\
15	0.00262367225184857\\
18	0.0025473221510768\\
21	0.00250342665221285\\
24	0.00248242742101104\\
27	0.00246164234882682\\
30	0.00244721629319107\\
};
\addlegendentry{$d/2r = 0.5$}

\addplot [color=mycolor2, line width=1.0pt]
  table[row sep=crcr]{%
1	1\\
2	0.138572842126947\\
3	0.123515509331121\\
4	0.101467301844314\\
5	0.0971682076843723\\
6	0.0910811274776843\\ 
7	0.0868992883149686\\
8	0.0861848932643529\\
9	0.08457458079868\\
10	0.0839773758259833\\
11	0.0830078474854945\\
12	0.0828103557500083\\
15	0.0815348073628744\\
18	0.0808754765098115\\
21	0.0804886066439046\\
24	0.0803008576038226\\
27	0.080114046784523\\
30	0.0799834056452741\\
};
\addlegendentry{$d/2r = 1.0$}

\addplot [color=mycolor3, line width=1.0pt]
  table[row sep=crcr]{%
1	1\\
2	0.292793049639804\\
3	0.275975318082198\\
4	0.249297583809281\\
5	0.243930792613369\\
6	0.235965775384912\\
7	0.230581874771191\\
8	0.229591023402445\\
9	0.227494463189831\\
10	0.226665431262875\\
11	0.225325571415376\\
12	0.225093046722408\\
15	0.223344025302189\\
18	0.222433432401212\\
21	0.221896976044152\\
24	0.221635750697412\\
27	0.221375751968901\\
30	0.221193658151067\\
};
\addlegendentry{$d/2r = 1.5$}

\addplot [color=mycolor4, line width=1.0pt]
  table[row sep=crcr]{%
1	1\\
2	0.41773738018898\\
3	0.40216442554857\\
4	0.376720601810441\\
5	0.371546252807915\\
6	0.363731767725114\\
7	0.358487005553923\\
8	0.357493603692921\\
9	0.355443285356767\\
10	0.354612970315219\\
11	0.353272832443005\\
12	0.353054966487674\\
15	0.35131115671816\\
18	0.350400925404582\\
21	0.349863905630704\\
24	0.349602075136963\\
27	0.349341463607397\\
30	0.34915884310199\\
};
\addlegendentry{$d/2r = 2.0$}

\addplot [color=mycolor5, line width=1.0pt]
  table[row sep=crcr]{%
1	1\\
2	0.511137672524018\\
3	0.497303214709367\\
4	0.47438215212293\\
5	0.469697372806201\\
6	0.462563487228047\\
7	0.457792696560457\\
8	0.456876432181875\\
9	0.455008042836332\\
10	0.454242776207689\\
11	0.453008327426233\\
12	0.452814025008376\\
15	0.451210261677647\\
18	0.450372115786245\\
21	0.449877289541121\\
24	0.449635885816378\\
27	0.449395604513768\\
30	0.449227188032362\\
};
\addlegendentry{$d/2r = 2.5$}

\addplot [color=mycolor6, line width=1.0pt]
  table[row sep=crcr]{%
1	1\\
2	0.581109953670194\\
3	0.568879006040967\\
4	0.548458895496317\\
5	0.544273857051004\\
6	0.537871895077832\\
7	0.533599315029172\\
8	0.53277238198131\\
9	0.531097466340443\\
10	0.530407145898174\\
11	0.529293917977566\\
12	0.52912185149344\\
15	0.527676848341148\\
18	0.526921171374842\\
21	0.526474867710589\\
24	0.526257063798843\\
27	0.526040272628866\\
30	0.52588829990343\\
};
\addlegendentry{$d/2r = 3.0$}

\addplot [color=mycolor7, line width=1.0pt]
  table[row sep=crcr]{%
1	1\\
2	0.634651771476399\\
3	0.62377869632602\\
4	0.605541140216459\\
5	0.601797235206787\\
6	0.596054228516361\\
7	0.592226267391555\\
8	0.591481888937974\\
9	0.58998038145576\\
10	0.589359177342937\\
11	0.588357578962061\\
12	0.588204488165075\\
15	0.586905086348096\\
18	0.586225279948148\\
21	0.585823695296751\\
24	0.585627675529162\\
27	0.585432567474739\\
30	0.585295783488426\\
};
\addlegendentry{$d/2r = 3.5$}

\addplot [color=mycolor8, line width=1.0pt]
  table[row sep=crcr]{%
1	1\\
2	0.676608301052411\\
3	0.666863097118882\\
4	0.650467837161848\\
5	0.647098517069236\\
6	0.641920780889651\\
7	0.638472473237879\\
8	0.637799850921291\\
9	0.636446740548105\\
10	0.635885542759751\\
11	0.634980793558532\\
12	0.634843518801896\\
15	0.633670178760799\\
18	0.633056163522047\\
21	0.632693390887289\\
24	0.632516292584042\\
27	0.632340018229279\\
30	0.632216431278762\\
};
\addlegendentry{$d/2r = 4.0$}
 
\end{axis}

\node at(\plotxtwo+\plotwidth-0.1in,\plotyone-0.25in) {\scriptsize $\times$ 1.15};

}

\end{tikzpicture}%
\caption{Least eigenvalue of $\wt{\mb G}$ with Gaussian motifs on hexagonal lattice. We show an example image of local features which are placed on the lattice locations (left), and calculate the least eigenvalue with varying number of motifs and distance-to-diameter ratio (right). When the motifs are highly overlapping $d/2r = 0.5$, then $\wt{\mb G}$ is almost rank-deficient; when $d/2r\geq 1$, then $\wt{\mb G}$ is stably full rank regardless of number of motifs. The result remains almost identical when the lattice is of other form such as rectangular grid, we therefore consider  the distance-to-diameter ratio is the determining factor for injectivity of line projections  even in more general settings. }\label{fig:gaussian-on-grid}
\end{figure}

\subsubsection{Injectivity of line projection of multiple motifs with minimum separation}

To extend the study of the coherence of matrix $\mb G$ to samples that contain $k > 2$ motifs $\mb D$. We first investigate a model configuration whose motif centers are allocated on a hexagonal lattice with edges of length $d$. It turns out that the smallest eigenvalue of an approximation $\mb G$  with respect to the locations $\set{\mb w_1,\ldots,\mb w_k}$ is largely determined by the distance-to-diameter ratio $d/2r$, and depends only weakly on the total number of motifs. 

In \Cref{fig:gaussian-on-grid}, we calculate an approximation of $\E_\theta\mb G$ with $\wt{\mb G}$, where 
\begin{align}
	\wt{\mb G}_{ij} = (1+\norm{\mb w_i-\mb w_j}2^2/4r^2)^{-1/2} 
\end{align}
is obtained from the upper bound in \Cref{lem:coherence_gauss} with motifs being the Gaussian function of deviation $r$ placed on hexagonal lattice. We show that when these motifs are highly overlapping with distance-to-diameter ratio $d/2r = 0.5$, the least eigenvalue of $\wt{\mb G}$ is very close to zero and the matrix is nearly rank-deficient; when the motifs are  separated, say $d/2r\geq 1$, the least eigenvalue of $\wt{\mb G}$ is steadily larger then zero and approaches one as the ratio $d/2r$ increases. Interestingly, in our experiments the least eigenvalue does not depend strongly on the number of motifs, suggesting that the distance-to-diameter ratio is the dominant factor for injectivity of line projections on motifs with hexagonal placement\footnote{The result of $\lambda_{\mr{min}}(\wt{\mb G})$ remains almost identical with other dense  motif allocation on lattice such as rectangular grid.}.  Since the hexagonal configuration is the densest circle packing on a plane \cite{toth2014regular}, we suspect that $\lambda_{\mr{min}}(\E_\theta\mb G)$ is also determined by the ratio $d/2r$ for every configurations satisfying the minimum separation property.

This conjecture gains more ground when viewing this problem from the point source super-resolution perspective \cite{candes2014towards}. It is known that an image consisting of point measures $\mb x = \sum_{i}\alpha_i\mb\delta_{\mb w_i}$  can be stably recovered from its low frequency information (with frequency cutoff $f_c$) whenever the point sources have minimum separation $d > C/f_c$ for some constant $C$, regardless of the number of such point measures in $\mb x$. In our scenario, we will show that the expected line projection $\E_\theta\mc L_\theta^*\mc L_\theta$ is also a low-pass filter; and since the local features $\mb D$ is also often consists of low frequency components, our line projections $\mc L_\Theta[\mb D*\mb X]$ can be modeled as the low-pass measurements from sparse map $\mb X$, implying stable and efficient sparse reconstruction is possible as long as $\mb X$ is enough separated under infinitely many line measurements of all angles. 

\begin{lemma}\label{lem:line_project_lowpass}\emph{[Lowpass property of line projections]} Suppose $\mb D$ is two-dimensional Gaussian of covariance $r^2\mb I$ with $\norm{\mc L_0[\mb D]}{L^2} = 1$ and $\mb X$ is finite summation of Dirac measure.  If $\theta$ is uniformly random, then $\E_\theta\mb D*\mc L_\theta^*\mc L_\theta[\mb D*\mb X]$ is a low-pass filter $\mc K$ on $\mb X$ with cut-off frequency $f_c$ satisfies
\begin{align}\label{eqn:cut-off-line-projection}
	f_c  = \tfrac1r\cdot \min\set{2r^2\eps^{-1},\, \sqrt{\abs{\log\paren{8r^2\eps^{-1}}}}+0.2}   
\end{align}
in a sense that $\max_{\norm{\mb\xi}2\geq f_c}\abs{\mc F_2\set{\mc K}(\mb\xi)} \leq \eps $.
	
\end{lemma}

\begin{proof} We start with restating projection slice theorem as $\mc F_{1}\mc L_\theta[\mb Y] = \mc S_\theta[\mc F_2 \mb Y]$, where $\mc F_1$, $\mc F_2$ are unitary Fourier transform in one, two-dimensional Euclidean space respectively,  $\mc S_\theta$ is the slice operator defined as $\mc S_{\theta}[\mb Y](r) = \mb Y(r\mb u_\theta^\perp)$ \cite{helgason2010integral}.

Notice that $\mb Y = \mb D*\mb X\in L^1\cap L^2(\R^2)$ therefore its Fourier transform is well defined, we expand $\mc L_\theta^*\mc L_\theta$ in Fourier domain with slice operator $\mc S_\theta$, write $\wh{\mb Y} = \mc F_2\mb Y$. and derive
	\begin{align} 
	&\E_\theta\mc L_\theta^*\mc L_\theta[\mb Y](\mb w) \notag \\
	&\quad =\, \E_\theta \mc F_2^*\mc S_\theta^*\mc F_1^{-1*}\mc F_1^{-1}\mc S_\theta\mc F_2\mb Y(\mb w) \notag \\
	&\quad=\, \E_\theta \textstyle\int_{\mb\xi\in\R^2} \exp\paren{j2\pi\innerprod{\mb\xi}{\mb w}}\cdot \mc S_\theta^*[\mc S_\theta[\wh{\mb Y}]](\mb\xi)\,d\mb\xi \notag \\
	&\quad =\, \E_\theta \textstyle\int_{t\in\R}\mc S_\theta[\exp\paren{j2\pi\innerprod{\cdot}{\mb w }}](t) \cdot \mc S_\theta[\wh{\mb Y}](t) \, dt \notag \\
	&\quad = \, \tfrac1{2\pi} \textstyle\int_{\theta = 0}^{2\pi} \textstyle\int_{t\in\R}\exp(j2\pi t\innerprod{\mb u_\theta^\perp}{\mb w} ) \cdot \wh{\mb Y}(t\mb u_\theta^\perp)\, dt\,d\theta\notag \\
	&\quad = \, \tfrac{2}{2\pi}\textstyle\int_{\theta=0}^{2\pi}\int_{t\geq 0} \exp(j2\pi t\innerprod{\mb u_\theta^\perp}{\mb w} ) \cdot \wh{\mb Y}(t\mb u_\theta^\perp)\, dt\, d\theta  \notag \\    
	&\quad = \,\textstyle\int_{\mb \xi\in\R^2}\exp(j2\pi t\innerprod{\mb \xi}{\mb w})\cdot\paren{\tfrac{1}{\pi\norm{\mb \xi}2}}\cdot \wh{\mb Y}(\mb \xi)\, d\mb \xi  \notag \\
	&\quad = \, \paren{\mc F_2^{-1}\Brac{\tfrac{1}{\pi\norm{\mb \xi}2}} * \mb Y} (\mb w), \label{eqn:low_pass_line_project}     
\end{align}  
where the third equality is derived from definition of adjoint operator, sixth equality is by coordinate transformation from polar to Cartesian, and the last equality is by convolution theorem. Hence we conclude that $\E_\theta\mc L_\theta^*\mc L_\theta[\mb Y]$ is the convolution between $\mb Y$ and a lowpass kernel with spectrum decay rate $\norm{\mb \xi}2^{-1}$. 
 
When $\norm{\mc L_0[\mb D]}{L^2}=1$ and is a Gaussian function with deviation $r$, then   $\mb D(\mb w) = \tfrac{\sqrt{2r\sqrt{\pi}}}{2\pi r^2}\exp\paren{-\tfrac{\norm{\mb w}2^2}{2r^2}}$ with Fourier domain expression as $\mc F_2\set{\mb D}(\mb\xi) = \sqrt{2r\sqrt{\pi}}\exp(-2\pi^2r^2\norm{\mb \xi}2^2)$. Combine with \eqref{eqn:low_pass_line_project}, the spectrum of $\E_\theta\mb D*\mc L_\theta^*\mc L_\theta[\mb D*\cdot\,]$ becomes
\begin{align}
	&\mc F_2\{\E_\theta\mb D*\mc L_\theta^*\mc L_\theta[\mb D*\mb X]\}(\mb\xi) \notag \\
	&\qquad = \tfrac{2r}{\sqrt{\pi} \norm{\mb\xi}2}\exp(-4\pi^2 r^2\norm{\mb \xi}2^2)\cdot\mc F_2\set{\mb X}(\mb\xi). \notag \\
	&\qquad = \mc F_2\set{\mc K}(\mb\xi) \cdot\mc F_2\set{\mb X}(\mb\xi)
\end{align}
Plug in \eqref{eqn:cut-off-line-projection}, when $\norm{\mb\xi}2\geq \tfrac{2r}{\eps}$ then clearly $\abs{\mc F_2\set{\mc K}(\mb\xi)} \leq \eps$. Lastly for the other lower bound $\norm{\mb\xi}2 \geq \tfrac1r\paren{\sqrt{\abs{\log(8r^2\eps^{-1})}} + 0.2 }$, we calculate
\begin{align} 
	\abs{\mc F_2\set{\mc K}(\mb\xi)} &\leq \tfrac{2 r^2}{0.2\sqrt\pi}\cdot\exp\paren{-4\pi^2\abs{\log(8r^2\eps^{-1})}} \notag \\
	&\leq \tfrac{2r^2}{0.2\sqrt\pi}\cdot\tfrac{1}{8r^2\eps^{-1}} \leq \eps.   
\end{align}  
\end{proof}

\begin{remark} When radius of motif is sufficiently large, then the cut-off frequency $f_c$ is dominated by the cut-off frequency of motif, roughly $C/r$, and is sufficient to recover its locations as long as the separation $d$ satisfies $d > C'r$ (reflects the observation of \Cref{fig:gaussian-on-grid}). In cases with small (pointy) $\mb D$, the cut-off  frequency is mainly determined by the low-pass property of line projection, which requires minimum separation $d > C\eps/r$ for exact reconstruction. 
\end{remark}
Finally, base on \cite{candes2014towards}, when the separation condition is ensured, the image of separated discs can be recovered from infinitely many line projections via total variation minimization (or $\ell^1$ when $\mb X_0$ on discrete grid), regardless of number of discs.	

\begin{figure}[t!] 
	\centering
	\input{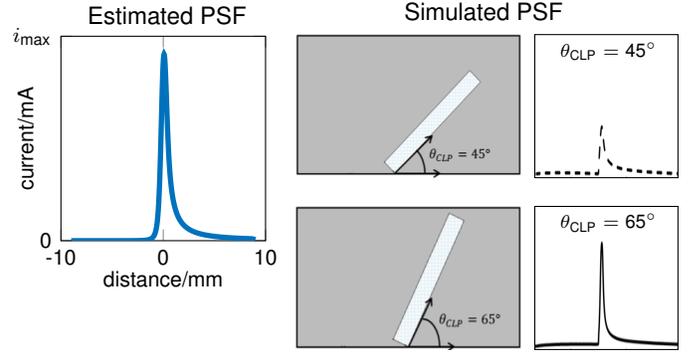}
	\caption{The point spread function of line probe. The PSF of line probe is skewed in the probe sweeping direction. We show an estimated PSF with close form used for reconstruction (left); and the software (Lab\acr{VIEW}) simulated PSF whose shape and intensity changes as the contacting angle varies (right).  }\label{fig:psf} 
\end{figure}

\subsection{Obstacles of image reconstruction from line scans} \label{sec:obstacles}
Besides the apparent nonideality of coherence of line scan measurements which is not CS theoretical optimal, this specific sampling method and its corresponding hardware limitations causes other practical nuisances during image reconstruction.

\paragraph{High coherence of line scans} To show the coherence is a cause for concern, we rewrite the linear operator $\mc L_\Theta[\mb D*\,\cdot\,]$ as $\mb A$, and consider the nonnegative Lasso
\begin{align}\label{eqn:lasso_exp}
	\min_{\mb X\geq 0} \lambda\norm{\mb X}1 + \tfrac12\norm{\mb A[\mb X] - \mb R}2^2
\end{align}
using the observed signal $\mb R = \mb A[\mb X_0]$ and linear, column normalized and coherent sampling method $\mb A$. Denote $\Omega$ as the support set of solution of \eqref{eqn:lasso_exp},  write $\mb A_{\Omega}$ as the submatrix of $\mb A$ restricted on columns of support $\Omega$,  the unique solution $\mb X$ of program \eqref{eqn:lasso_exp} (provided if $\mb A_{\Omega}$ is injective) can be written as
\begin{align}
	\begin{cases}
		\mb X_{ij} = \brac{\mb X_{0ij} - \lambda(\mb A_{\Omega}^*\mb A_{\Omega})^{-1}\1}_+&\quad {\mb w_{ij}}\in\Omega \\
		\mb X_{ij} = 0 &\quad{\mb w_{ij}\not\in\Omega}.
	\end{cases}  
\end{align}  
When $\mb A$ is coherent, columns of $\mb A$ have large inner product, implies many entries of the  matrix $\mb A_{\Omega}^*\mb A_{\Omega}$ have large, positive off-diagonal entries close to its diagonals.  When the sparse penalty $\lambda$ is large in \eqref{eqn:lasso_exp}, its solution will have incorrect relative magnitudes since $\mb A_{\Omega}^*\mb A_{\Omega}$ is not close to identity matrix as conventional CS measurements \cite{candes2005decoding}. When $\lambda$ is small, the solution of program will be highly sensitive to noise, occasionally lead to incorrect results.

\paragraph{Incomplete information of PSF of line scans} Another layer of complexity for line probe scans is the difficulty to correctly identify its PSF due to hardware limitations, especially when operating line scans in nanoscale. For instance in \Cref{fig:psf}, we show if the contacting angle between the probe and the sample varies, the corresponding PSF changes drastically in both the peak magnitude and the shape. It turns out that even with seemingly small changes of probe condition, the corresponding PSF can be inevitably variated.


\section{Reconstruction from line scans}\label{sec:algorithm}
In this section, we introduce an algorithm for reconstructing SECM images from line scans. In all following experiments, we consider a representative class of images $\mb Y$ characterized by superposing reactive species $\mb D$ at locations $\mc W = \set{\mb w_1,\ldots,\mb w_{\abs{\mc W}}}\subset \R^2$ with intensities $\set{\alpha_1,\ldots,\alpha_{\abs{\mc W}}}\subset\R_+ $. Define the activation map $\mb X_0$ as sum of Dirac measure at $\mc W$, then $\mb Y$ can simply be written as convolution between $\mb D$ and $\mb X_0$:
\begin{equation}\label{eqn:Y}
	\mb Y= \mb D*\mb X_0 = \textstyle\sum_{j=1}^{\abs{\mc W}}{\alpha_j\mb D*\mb\delta_{\mb w_j}}.
\end{equation}
The imaging reconstruction problem then can be cast as finding the best fitting sparse map  $\wh{\mb X}$ from line scans $\mb R = \mc S\{ \mb \Psi*\mc L_\Theta[\mb Y] \} $, and the reconstructed image is simply $\mb D*\wh{\mb X}$. Since all associated operations on $\mb X_0$ (convolution with $\mb D,\mb\psi$ and line projection $\mc L_\Theta$) are all linear, this becomes a sparse estimation problem, which can be solved via the Lasso. In practice, due the resolution limit of probe and the sampling operation $\mc S$, we do not aiming to find exact $\mb X$ in a continuous space. Instead, we will solve the discretized version of this sparse recovery problem, which assume $\mb X$ resides on a grid. As such, the associated Lasso problem can be written as:
\begin{align}\label{eqn:lasso_l1}
	\min_{\mb X\geq 0} \lambda \textstyle\sum_{ij}\mb X_{ij} + \tfrac12\norm{\mb R - \mc S\{\mb\Psi * \mc L_\Theta[\mb D*\mb X]\}}2^2.
\end{align} 

%

\subsection{Sparse recovery with Lasso from line projections}   
In light of \Cref{sec:line-projection-low-pass}, the measurement performance using infinitely many line scans is almost dependent only on the distance-to-diameter ratio of the local features. Since in practice, only finite number line scan is available, we want to study how many line scans will be sufficient for efficient and exact sparse image reconstruction. We do this by studying the performance of algorithm \eqref{eqn:lasso_l1}  while assuming the line scan are idealized where the PSF is ideally all-pass in the sense that $\mb\psi=\mb\delta$.
  
\begin{figure}[t!]
\centering
\begin{tikzpicture}
\newcommand{\plotxone}{0in}
\newcommand{\plotxtwo}{0.25\textwidth}
\newcommand{\plotyone}{0in}
\newcommand{\plotytwo}{0.26\textwidth}
\newcommand{\plotwidth}{0.19\textwidth}
\newcommand{\plotheight}{0.18\textwidth}

\definecolor{mycolor1}{rgb}{0.00000,0.44700,0.74100}%
\definecolor{mycolor2}{rgb}{0.85000,0.32500,0.09800}%
\definecolor{mycolor3}{rgb}{0.92900,0.69400,0.12500}%
\definecolor{mycolor4}{rgb}{0.49400,0.18400,0.55600}%
	
\sffamily{ 

\begin{axis}[
width=\plotwidth,
height=\plotheight,
at={(\plotxtwo,\plotytwo)},
scale only axis,
axis on top,
xmin=8,
xmax=168,
ymin=0,
ymax=4000,
xtick={16,64,112,160},
ytick={140,500,3600},
xticklabels={\scriptsize 16,\scriptsize 64,\scriptsize 112,\scriptsize 160},
yticklabels={\scriptsize 140,\scriptsize 500,\scriptsize 3.6k},
axis background/.style={fill=white},
legend image post style={scale=0.5},
legend style={at={(0.82in,\plotheight/2+0.45in)},legend cell align=left, font=\scriptsize, align=left, draw=none, fill=none}
]

\addplot [color=mycolor1, line width=1.5pt]
  table[row sep=crcr]{%
16  3600\\
32  3600\\
48  3600\\
64  3600\\
80  3600\\
96  3600\\
112 3600\\
128 3600\\
144 3600\\
160 3600\\
};
\addlegendentry{{point probe}}

\addplot [color=mycolor2, line width=1.5pt]
  table[row sep=crcr]{%
16  140\\
32  180\\
48  220\\
64  260\\
80  300\\
96  340\\
112 380\\
128 420\\
144 460\\
160 500\\
};
\addlegendentry{{line probe}}
\end{axis}

\begin{axis}[
width=\plotwidth,
height=\plotheight,
at={(\plotxtwo,\plotyone)},
scale only axis,
axis on top,
xmin=15,
xmax=125,
ymin=0,
ymax=2667,
xtick={20,40,60,80,100,120},
ytick={140,400,790,2400},
xticklabels={\scriptsize{20},\scriptsize{40},\scriptsize{60},\scriptsize{80},\scriptsize{100},\scriptsize{120}},
yticklabels={\scriptsize{140},\scriptsize{400},
    \scriptsize{790},
    \scriptsize{2.4k}},
axis background/.style={fill=white},
legend image post style={scale=0.5},
legend style={at={(0.82in,\plotheight/2+0.6in)},legend cell align=left, font=\scriptsize, align=left, draw=none, fill=none} 
] 

\addplot [color=mycolor1, line width=1.5pt]
  table[row sep=crcr]{%
20  400\\
30  600\\
40  800\\
50  1000\\
60  1200\\
70  1400\\
80  1600\\
90  1800\\
100 2000\\
110 2200\\
120 2400\\
};
\addlegendentry{{point probe}}

\addplot [color=mycolor2, line width=1.5pt]
  table[row sep=crcr]{%
20  140\\
30  206\\
40  272\\
50  337\\
60  402\\
70  466\\
80  531\\
90  595\\
100 660\\
110 724\\
120 788\\
};
\addlegendentry{{line probe}}
\end{axis}

\begin{axis}[
width=\plotwidth,
height=\plotheight,
at={(\plotxone,\plotytwo)},
scale only axis,
axis on top, 
axis line style={draw=none},
tick style={draw=none},
xmin=0, 
xmax=180,
ymin=1.7,
ymax=22.5, 
xtick={16,64,112,160},
ytick={3,7,11,15,19,21}, 
xticklabels={\scriptsize{16},\scriptsize{64},\scriptsize{{112}},\scriptsize{160}},
yticklabels={\scriptsize{3},\scriptsize{{7}},\scriptsize{11},\scriptsize{15},\scriptsize{19},\scriptsize{21}},
axis background/.style={fill=white},
legend image post style={scale=0.5},
legend style={at={(1.0in,1.33in)}, legend cell align=left, align=left, draw=none, fill=none, font=\scriptsize}
]
\addplot [forget plot] graphics [xmin=0, xmax=180, ymin=1.7, ymax=22.5] {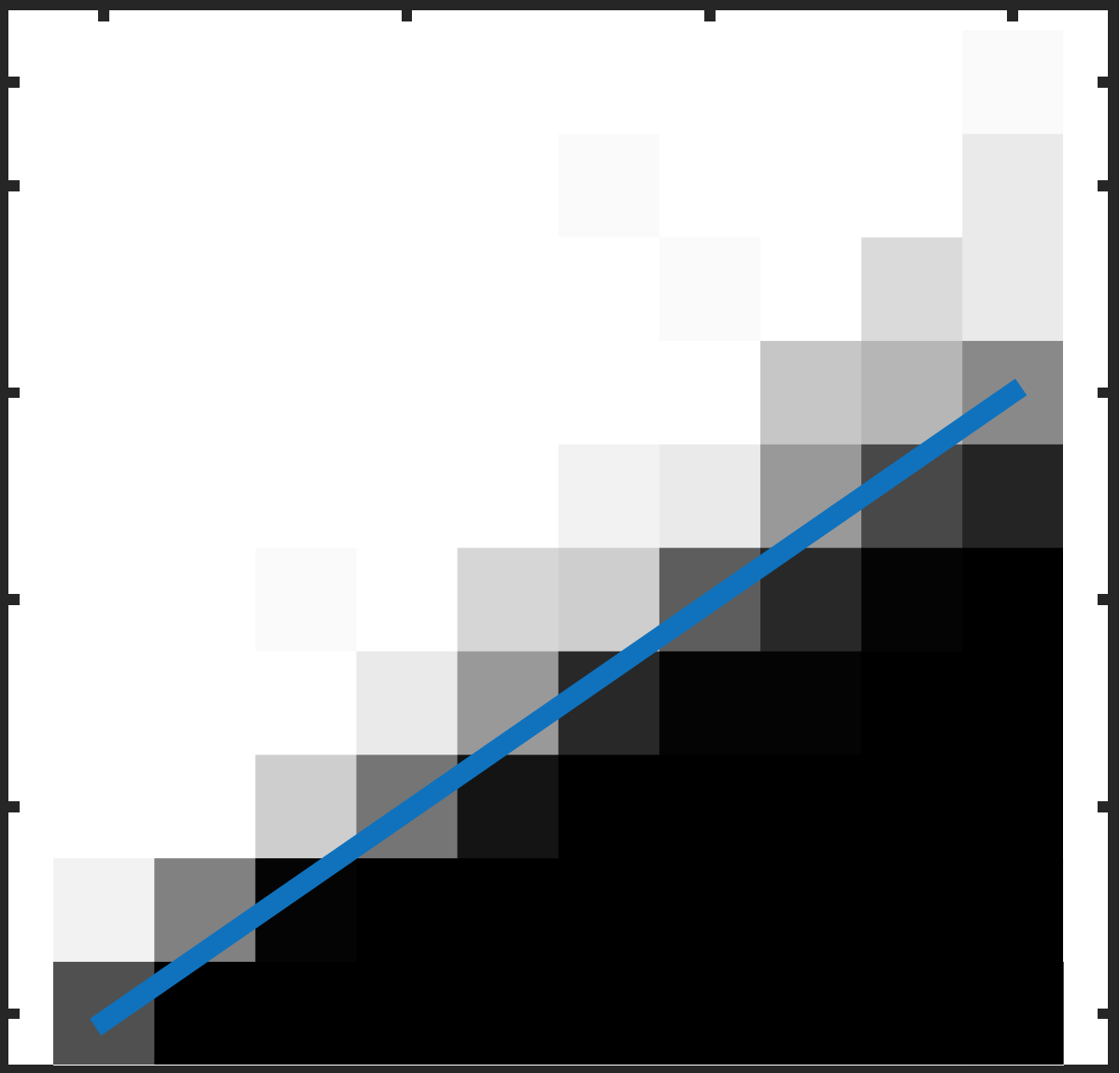};

\addplot [color=mycolor1, line width=1.5pt]
  table[row sep=crcr]{%
  0 0\\
};
\addlegendentry{$\approx$ 50\% recover}
\end{axis}

\begin{axis}[
width=\plotwidth,
height=\plotheight,
at={(\plotxone,\plotyone)},
scale only axis,
axis on top,
axis line style={draw=none},
tick style={draw=none},
xmin=10,
xmax=130,
ymin=4.3,
ymax=16,
xtick={20,40,60,80,100,120},
ytick={5,7,9,11,13,15},
xticklabels={\scriptsize{20},\scriptsize{40},\scriptsize{60},\scriptsize{80},\scriptsize{100},\scriptsize{120}},
yticklabels={\scriptsize{5},\scriptsize{7},\scriptsize{9},\scriptsize{11},\scriptsize{13},\scriptsize{15}},
axis background/.style={fill=white},
legend image post style={scale=0.5},
legend style={at={(1.1in,1.7in)}, legend cell align=left, align=left, draw=none, fill=none, font=\scriptsize}
]
\addplot [forget plot] graphics [xmin=10, xmax=130, ymin=4.3, ymax=16] {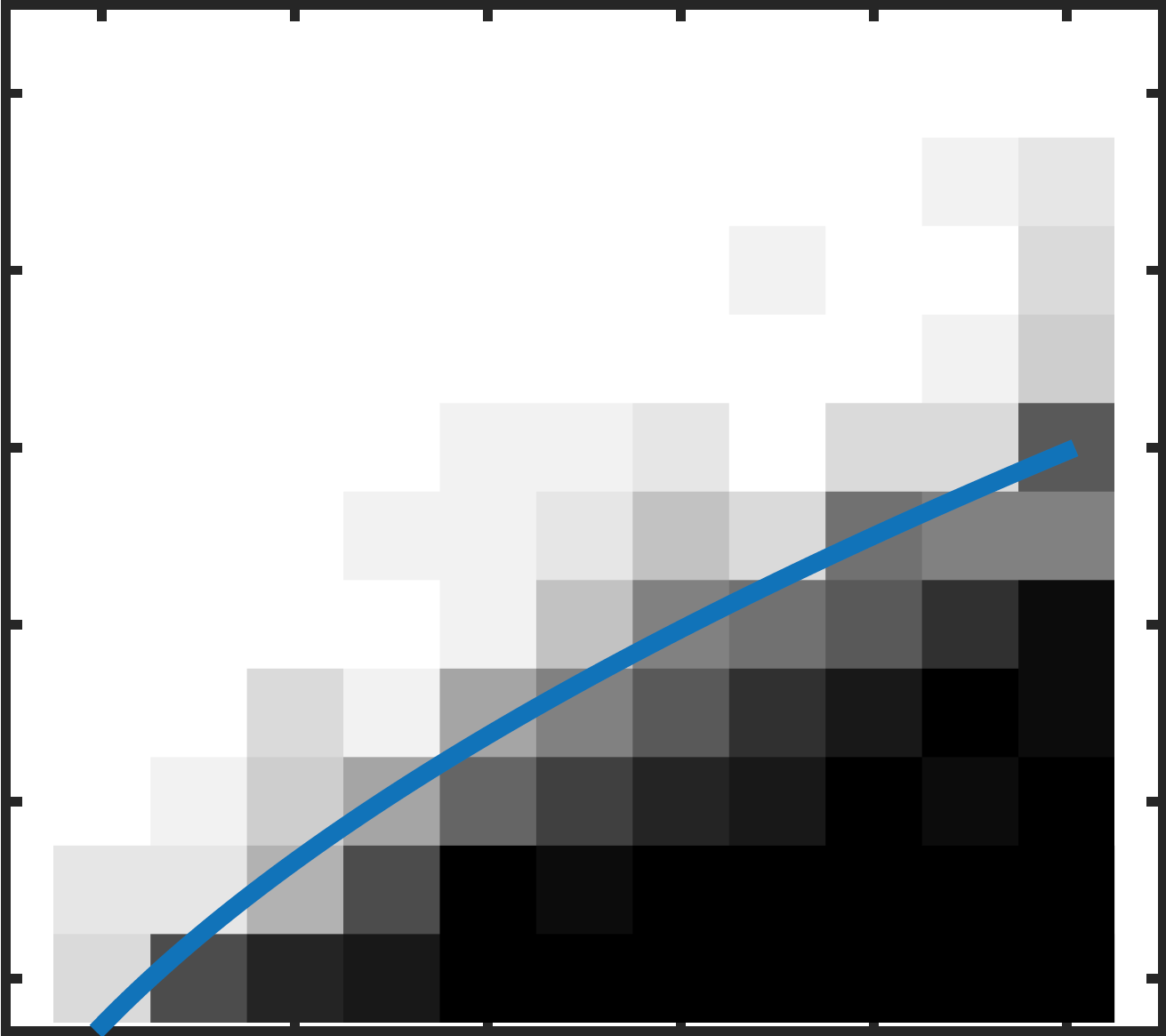}; 

\addplot [color=mycolor1, line width=1.5pt]
  table[row sep=crcr]{%
  0 0\\
};
\addlegendentry{$\approx$ 50\% recover}
\end{axis}

\node at (\plotxone/2+\plotxtwo/2+\plotwidth/2,\plotytwo+\plotheight + 0.1in) {Fixed image size};
\node at (\plotxtwo/2+\plotxone/2+\plotwidth/2,\plotyone+\plotheight + 0.1in) {Fixed image density};

\node at (\plotxone+\plotwidth/2,\plotytwo-0.20in) {\footnotesize number of discs};

\node at (\plotxtwo+\plotwidth/2,\plotytwo-0.20in) {\footnotesize number of discs};

\node[rotate=90] at (\plotxone-0.25in,\plotytwo+\plotheight/2) {\footnotesize number of lines};

\node[rotate=90] at (\plotxtwo-0.27in,\plotytwo+\plotheight/2) {\footnotesize number of lines};

\node at (\plotxone+\plotwidth/2,\plotyone-0.2in) {\footnotesize number of discs};

\node at (\plotxtwo+\plotwidth/2,\plotyone-0.2in) {\footnotesize number of discs};

\node[rotate=90] at (\plotxone-0.25in,\plotyone+\plotheight/2) {\footnotesize number of samples};

\node[rotate=90] at (\plotxtwo-0.27in,\plotyone+\plotheight/2) {\footnotesize  number of samples };

}

\end{tikzpicture}
\caption{Phase transition \cite{o2018scanning} of fixed image size (top) and fixed density (bot) on support recovery with Lasso. In each experiments, $d/2r \geq 1$ is ensured. In either cases, the phase transitions (left) show the number of samples required is almost linearly proportional to the number of discs for exact reconstruction.  And the the advancement of scanning efficiency (right) is presented in comparison with the point probe scans. For the fixed size case, we let (image area)/(disc area) $\approx 1200$; for the fixed density case, we let density $\approx (1/6)\cdot$(max density).}\label{fig:PT}
\end{figure}
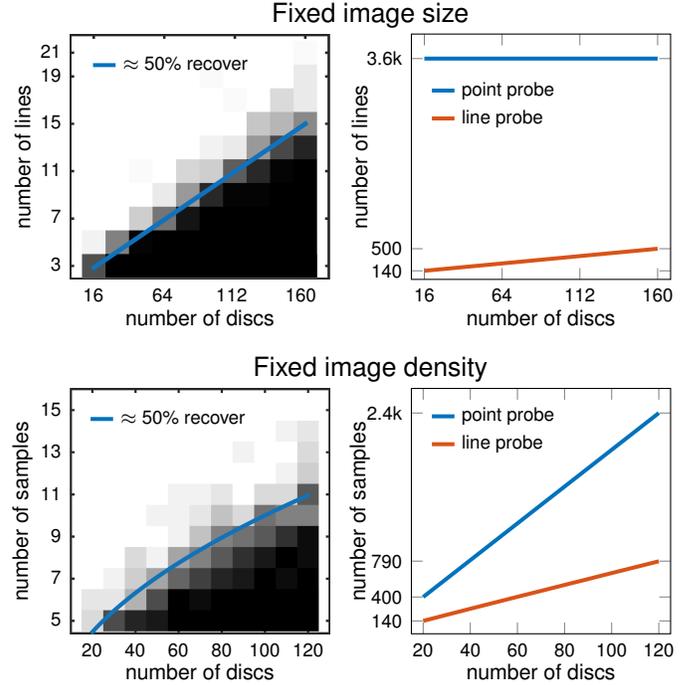    
  
\Cref{fig:PT} shows the reconstruction performance from line scans with varying number of lines used and number of discs in the target image $\mb Y$. Each image $\mb Y$ is  generated by randomly populating the discs of size $r$ while satisfying $d/2r \geq 1$ via rejection sampling, and the scan angles are also uniformly random chosen. Here, two experiment settings are presented. The first is assumed that the imaging area of line scan is fixed (so the density increases linearly with more discs) and the second is considering the cases where the density is a constant (so the imaging area is proportional to the disc amount).  In the phase transition (PT) image (\Cref{fig:PT}, left), each pixel represents the average of 50 experiments; and in each experiment, given random image $\mb Y$ and its line scans of randomly chosen angles, if solving \eqref{eqn:lasso_l1} correctly identify the support map of $\mb Y$, then the algorithm succeeds, and vice versa. It shows clear transition lines in both PT images, and the comparison of scanning time between line/point probes shows clear improvement of scanning efficiency.

  Interestingly if we compare the result with CS theory, which asserts the number measurement of samples required is close to linear proportional to signal sparsity; here, though the line scans are not CS-optimal, both PT images exhibits similar phenomenon. When the image size is fixed (up), total number of samples $m$ is proportional to the line count $N$, with PT transition line showing linear proportionality between number of line scans and discs $N\propto k$, gives $m\propto k$; on the other hand, when the image density is fixed, the number of samples $m$ is proportional to $N \times \sqrt k$\footnote{With fixed density, imaging area is proportional to disc count, and the number of samples is $(\text{line count})\times\sqrt{(\text{imaging area})} = N \times \sqrt k$.} while the transition line in PT is showing  $N \propto\sqrt k$,  again suggests linear proportionality between the number of measurements and sparsity would be  $m\propto N\sqrt k\propto k$. To wrap up, these experimental results hinted that if minimum separation of discs are ensured, then to ensure exact signal reconstruction with efficient algorithm, the number of samples required is approximately linearly proportional to the sparsity of image. 

Finally, to formally elucidate the sample time reduction from point probe to line scans,  we compare the consumed scanning time using different probes in both settings under specific scenarios. (\Cref{fig:PT}, right). In the fixed area experiment we let the image area be $3\!\times\!3\, \mr{mm}^2$ and the disc radius and the image resolution are both $50\mr{\mu m}$ (image area/disc area ratio around $1200$); for the fixed density we let all experiments have equal density $20$ $\text{discs}/\mr{mm}^2$ (nearly $1/6$ of maximum density in separating case) with same resolution. Both of the results show clear improvement of scanning efficiency, with reduction of scanning time by 3 to  10 times under these signal settings.

 In either case, line measurements are substantially more efficient than measurements with a point probe. Realizing this gain in practice requires us to modify the Lasso to cope with the following nonidealities: (i)  line scans are coherent, (ii) the PSF $\mb \psi$ is typically only partially known, and (iii) naive approaches to computing with line scans are inefficient when the target resolution is large. Below, we show how to address these issues, and give a complete reconstruction algorithm.

%

\subsection{Computation of line projection}
  
\subsubsection{Fast computation of discrete line projection}

The line projection of an image $\mb Y$ in direction of angle $\theta$ is equivalent to the line projection at  $0^\circ$ of clockwise rotated $\mb Y$  by angle $\theta$. This enables an efficient line projection computationally via fast image rotation with shear transform in Fourier domain \cite{larkin1997fast}.  

The clockwise rotation of image $\mb Y$ by angle $\theta$ is 
\begin{align}\label{eqn:rotate_fourier_1}
	\mr{Rot}_\theta\brac{\mb Y}(x,y) \,=\, \mb Y\paren{\begin{bmatrix}
		\cos\theta & -\sin\theta \\ \sin\theta & \cos\theta 
	\end{bmatrix}\begin{bmatrix}
		x \\ y
	\end{bmatrix}}   
\end{align} 
where the rotational matrix can be decomposed into three shear transforms  
\begin{align}
	\begin{bmatrix}
		\cos\theta & -\sin\theta \\ \sin\theta & \cos\theta 
	\end{bmatrix} = \begin{bmatrix}
 		1 & 0 \\ \tan\tfrac\theta2 & 1
 \end{bmatrix} \begin{bmatrix}
 		1 & -\sin\theta \\ 0 & 1
 \end{bmatrix} \begin{bmatrix}
 		1 & 0 \\ \tan\tfrac\theta2 & 1
 \end{bmatrix}; \notag 
\end{align}  
write both $x$, $y$-shear transforms as 
\begin{align}
	\mr{Shr}\text{-}\mr{x}_s[\mb Y](x,y) &= \mb Y(x+sy,y),\notag \\
	\mr{Shr}\text{-}\mr{y}_t[\mb Y](x,y) &= \mb Y(x,y+tx),\notag
\end{align}
then  
\begin{align}
	\mr{Rot}_\theta[\mb Y] = \mr{Shr}\text{-}\mr{y}_{\tan\tfrac{\theta}2} \circ \mr{Shr}\text{-}\mr{x}_{-\sin\theta} \circ \mr{Shr}\text{-}\mr{y}_{\tan\tfrac{\theta}2}\brac{\mb Y}. 
\end{align} 
Each of the shear transform can be efficiently computed in Fourier domain. Define 
\begin{align}\label{eqn:shear-fourier}
	\wh{\mc S}_{x,t}(u,y)= e^{j2\pi t y u},\quad \wh{\mc S}_{y,t}(x,v)= e^{j2\pi t x v}, 
\end{align}
and $\mc F_x$, $\mc F_y$ as $n$-DFT in $x$,$y$-domain, where
\begin{align}
	\mc F_x\{\mb Y\}(u,y) \,=\, \textstyle\sum_{x}\mb Y(x,y)e^{-j2\pi x u },  \\
	\mc F_y\{\mb Y\}(x,v) \,=\, \textstyle\sum_{y}\mb Y(x,y)e^{-j2\pi y v}. \label{eqn:dft-xy} 
\end{align}
From \eqref{eqn:shear-fourier}-\eqref{eqn:dft-xy}, the $y$-shearing transform can be written as
\begin{align}  
	\mb Y(x,y+tx) &= \textstyle\sum_{y'} \mb Y(x,y')\mb\delta(y'-tx-y) \notag \\
	& = \tfrac1n\textstyle\sum_{y'}\mb Y(x,y')\textstyle\sum_v e^{-j2\pi v(y'-tx - y)} \notag \\
	& = \tfrac1n\textstyle\sum_v\paren{\textstyle\sum_{y'} \mb Y(x,y')e^{-j2\pi v(y'-tx)}}e^{j2\pi v y} \notag \\
	& = \mc F_y^{-1}\big[\mc F_y\brac{\mb Y}\circ \wh{\mc S}_{y,t}\big];  
\end{align}
and $x$-shear transform likewise,  
\begin{align}\label{eqn:rotate_fourier_5}
	\mb Y(x + ty , y) = \mc F_x^{-1}\big[\mc F_x\brac{\mb Y}\circ\wh{\mc S}_{x,t} \big].
\end{align}
Combine \eqref{eqn:rotate_fourier_1}-\eqref{eqn:rotate_fourier_5}, we obtain a computational efficient algorithm for line projections \Cref{alg:line-project}.

\begin{algorithm}[h]\label{alg:line-project}
\myfont{
\caption{Fast computational discrete line projections}\label{alg:line-project}
\begin{algorithmic}
	\Require Discrete image $\mb Y\in\R^{n\times n}$, line scan angles $\set{\theta_1,\ldots,\theta_m}$.
	\For{$i=1,\dots, m$}
	\State $y$-shearing: $\mb Y \gets \mc F_y^{-1}\big[\mc F_y\brac{\mb Y} \circ \wh{\mc S}_{y,\tan(\theta_i/2)}\big]$;  
	\State $x$-shearing: $\mb Y \gets \mc F_x^{-1}\big[\mc F_x\brac{\mb Y} \circ \wh{\mc S}_{x,-\sin\theta_i}\big]$;
	\State $y$-shearing:  $\mb Y \gets \mc F_y^{-1}\big[\mc F_y\brac{\mb Y} \circ \wh{\mc S}_{y,\tan(\theta_i/2)}\big]$;  
	\For{$t = 1,\ldots,n$} 
	\State $\mb R_i(t) \gets \frac{1}{\sqrt m}\sum_{y}\mb Y(t,y)$;
	\EndFor
	\EndFor
	\Ensure Discrete lines $\mc L_\Theta[\mb Y] = \set{\mb R_1,\ldots,\mb R_m}\in\R^{n\times m}$
	\end{algorithmic}}
\end{algorithm}
Since the image $\mb Y$ is discrete, rotation will naturally incur interpolation error. To mitigate this effect, it is advised to limit the rotation operation to angle $\theta\in\brac{-45^\circ, 45^\circ}$ in \Cref{alg:line-project}, then flip the image vertically or horizontally to form the image rotated by $[-180^\circ,180^\circ)$.

Although the Fourier rotation method demands $O(n^2\log n)$ for computational time, which is slightly larger then the direct rotation $O(n^2)$, in practice we found Fourier rotation more appealing: its actual computational time is usually slightly better then other methods, since it gets around the problematic pixelated interpolation from direct rotation; and more importantly,  its adjoint is easy to calculate in a similarly explicit manner as well.

\subsubsection{Adjoint of line projection} 
 The adjoint operator\footnote{We invoke the canonical definition of inner product of $L^2$-space for both image and lines. For every images $\mb Y,\mb Y'\in L^2(\R^2)$, we define $\innerprod{\mb Y}{\mb Y'} = \int \mb Y(\mb w)\mb Y'(\mb w)\,d\mb w $; and for every lines $\wt{\mb R},\wt{\mb R}'\in L^2(\R\times [m])$, we define $\langle\wt{\mb R},\,\wt{\mb R}'\rangle = \sum_{i=1}^m\int\wt{\mb R}_i(t)\wt{\mb R}_i'(t)\, dt$.} of line projections $\mc L_\Theta^*:L^2([m]\times\R)\to L^2(\R^2)$ is deeply connected with the well-known tomography image reconstruction technique \emph{back projection}. The adjoint of a single line projection $\mc L^*_{\theta_i}:L^2(\R)\to L^2(\R^2)$ of scanning angle $\theta_i$ is exactly the back projection of a continuous line  $\wt{\mb R}_i$ which generates an image $\mc L_{\theta_i}^*[\wt{\mb R}_i]$ whose value over $\ell_{\theta_i,t}$ defined in in \eqref{eqn:line_set} is equivalent to $\wt{\mb R}_i(t)$:  
\begin{align}
	\mc L_{\theta_i}^*[\wt{\mb R}_i](\mb w) = \wt{\mb R}_i(t),\qquad \forall\,\mb w\in\ell_{\theta_i,t}, 
\end{align}
then incorporate with definition of $\ell_{\theta_i,t}$, we obtain a simpler form for $\mc L^*_{\theta_i}$ as 
\begin{align}\label{eqn:back_project_one_line}
	\mc L_{\theta_i}^*[\wt{\mb R}_i](\mb w) = \wt{\mb R}_i(\langle\mb u_{\theta_i}^\perp\mb w\rangle).
\end{align}
Extending the derivation of \eqref{eqn:back_project_one_line} to $m$-lines $\wt{\mb R}$, the back projection of $m$ angles $\mc L_\Theta^*$ on $\wt{\mb R}$ is the superposition of images from all $m$ back projected lines $\mc L_{\theta_i}^*[\wt{\mb R}_i]$ of different scanning angles:
\begin{align}\label{eqn:back_project}
	\mc L^*_{\Theta}[\wt{\mb R}](\mb w) &\,=\, \tfrac{1}{\sqrt m}\textstyle\sum_{i=1}^m\mc L^*_{\theta_i}[\wt{\mb R}_i](\mb w)\notag \\
	&\,=\,  \tfrac{1}{\sqrt m} \textstyle\sum_{i=1}^m\wt{\mb R}_i(\innerprod{\mb u_{\theta_i}^\perp}{\mb w}).
\end{align} 
In the following proposition, we show that the line projections defined in \eqref{eqn:back_project} is indeed the adjoint operator of line projections. 

\begin{figure}[t!]
	\centering	
 	\input{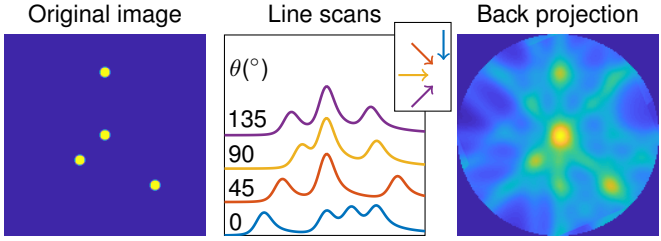}
	 
	\caption{ Back projection image from the scan lines. We demonstrate a simple example (left) where four discs are line projected with angles $\set{0^\circ,45^\circ,90^\circ,135^\circ}$ then undergo convolution with the simulated PSF (mid). Here, the arrows indicates the probe sweeping direction. The back projection image (right) is the superposition of back projection image of each line; and the back projection of  a single line $\mb R_\theta$ assigns value $\mb R_\theta(t)$ along the sweeping directions (arrows) onto the support $\ell_{\theta,t}$ for every $t$. }\label{fig:back_project}
\end{figure}

\begin{proposition}\label{prop:backproject_is_adjoint} The back projection $\mc L_\Theta^*$ in \eqref{eqn:back_project} is the adjoint of line projection $\mc L_\Theta$ in \eqref{eqn:line_project}, where 
\begin{align}\label{eqn:back_project_adjoint_of_line_project}
	\langle\wt{\mb R},\mc L_\Theta[\mb Y]\rangle = \langle\mc L_\Theta^*[\wt{\mb R}],\mb Y\rangle.
\end{align} 
\end{proposition}

\begin{proof} For any lines $\mb R\in L^2(\R\times[m])$, image $\mb Y\in L^2(\R^2)$, and any angles $\Theta = \set{\theta_1,\ldots,\theta_m}$,  
\begin{align} 
 	\langle\wt{\mb R},\mc L_{\Theta}[\mb Y]\rangle &\;=\;  \tfrac{1}{\sqrt m}\textstyle\sum_{i=1}^m\textstyle\int \wt{\mb R}_i(t)\mc L_{\theta_i}[\mb Y](t)\,dt\notag \\
 		&\;=\; \tfrac{1}{\sqrt m} \textstyle\sum_{i=1}^m \int \wt{\mb R}_i(t)\textstyle\int \mb Y(s\mb u_{\theta_i} + t\mb u_{\theta_i}^\perp)\,ds\,dt \notag \\
 		&\;=\; \tfrac{1}{\sqrt m}\textstyle\sum_{i=1}^m\textstyle\int \wt{\mb R}_i\paren{\innerprod{\mb w}{\mb u_{\theta_i}^\perp}}\mb Y(\mb w)\, d\mb w   \notag \\ 
 		&\;=\; \textstyle\int\big(\frac{1}{\sqrt m}\textstyle\sum_{i=1}^m\wt{\mb R}_i\paren{\innerprod{\mb w}{\mb u_{\theta_i}^\perp}}\big)\mb Y(\mb w)\,d\mb w \notag \\
 		&\;=\; \langle\mc L_\Theta^*[\wt{\mb R}],\mb Y\rangle. 
 \end{align}
 The first equality comes  from the definition of inner product in lines space; the second comes  from \eqref{eqn:line_project}; the third uses change of variable where $\mb w = s\mb u_\theta + t\mb u_\theta^\perp$ for every $\theta$; the fourth comes from linearity;  and the last equality from definition of inner product in image space.
 \end{proof}

\subsubsection{Fast computation of discrete back projection}
Similar to the line projection, the discrete back projection of a single line $\mb R_i\in\R^n$ at angle $\theta$ is the image $\mb Y_i = [\mb R_i,\mb R_i,\cdots,\mb R_i]\in\R^{n\times n}$ counterclockwise rotated by $\theta$, and the back projection of multiple lines is the sum of all such images, as shown in \Cref{fig:back_project}.  The discrete back projection thereby can be also calculated efficiently in Fourier domain, as presented in \Cref{alg:back-project}.

\begin{algorithm}[h]\myfont{
\caption{Fast computational discrete back projections}\label{alg:back-project}
\begin{algorithmic}
	\Require Discrete lines $\set{\mb R_1,\ldots,\mb R_m}\in\R^{n\times m}$, line scan angles $\set{\theta_1,\ldots,\theta_m}$.
    \State Initialize $\mb Y \gets \mb 0 \in\R^{n\times n}$; 
	\For{$i=1,\dots, m$}
		\For{$x = 1,\ldots,n$} 
		\State $\mb Y_i(x,:) \gets \frac{1}{\sqrt m}\mb R$;
		\EndFor 
	\State $y$-shearing: $\mb Y_i \gets \mc F_y^{-1}\big[\mc F_y\brac{\mb Y_i} \circ \wh{\mc S}_{y,-\tan(\theta_i/2)}\big]$;  
	\State $x$-shearing: $\mb Y_i \gets \mc F_x^{-1}\big[\mc F_x\brac{\mb Y_i} \circ \wh{\mc S}_{x,\sin\theta_i}\big]$;
	\State $y$-shearing:  $\mb Y_i\gets \mc F_y^{-1}\big[\mc F_y\brac{\mb Y_i} \circ \wh{\mc S}_{y,-\tan(\theta_i/2)}\big]$; 
	\State $\mb Y\gets \mb Y + \mb Y_i$  
	\EndFor
	\Ensure Discrete image $\mb Y\in\R^{n\times n}$
	\end{algorithmic}}
\end{algorithm}
\begin{remark} The discrete back projection from \Cref{alg:back-project} is the adjoint operator of discrete line projection from \Cref{alg:line-project}, which satisfies $\langle\wt{\mb R},\mc L_\Theta[\mb Y]\rangle = \langle\mc L_\Theta^*[\wt{\mb R}],\mb Y\rangle$.
	
\end{remark}

\subsection{Coping with nonidealities}
As aforementioned in \Cref{sec:obstacles}, the vanilla Lasso formulation in \eqref{eqn:lasso_l1} does not provide a convincing solution for practical problems in SECM with line probe, due to the high coherence of line projections and the nonidealities of PSF. These issues can be remedied by implementing well known techniques such as reweighting and blind calibration.

\subsubsection{Reweighting  Lasso for coherent measurements} \label{sec:alg-reweight}
 To cope with the coherence phenomenon, we adopt the reweighting scheme \cite{candes2008enhancing} by solving Lasso formulation \eqref{eqn:lasso_l1} multiple times while updating penalty variable $\mb\lambda$ in each iterate. At $k$-th iterate, the algorithm chooses the regularizer $\mb\lambda$ in \eqref{eqn:lasso_l1} base on the previous outcome of lasso solution $\mb X^{(k)}$, where
\begin{align}\label{eqn:lda_reweight}
	\mb \lambda_{ij}^{(k)}\gets C(\mb X^{(k-1)}_{ij}+ \eps)^{-1} 
\end{align}
Reweighting \cite{candes2008enhancing} is a technique in sparse recovery which is typically utilized for enhancing the sparsity regularizer, by solving Lasso formulation \eqref{eqn:lasso_l1} multiple times while updating penalty variable $\mb\lambda$ in each iterate. At $k$-th iterate, the algorithm chooses the regularizer $\mb\lambda^{(k)}_{ij}$ based on the previous outcome of lasso solution $\mb X^{(k)}$, where
\begin{align}\label{eqn:lda_reweight}
	\mb \lambda_{ij}^{(k)}\gets C(\mb X^{(k-1)}_{ij}+ \eps)^{-1} 
\end{align}
with $\eps$ being the machine precision constant and $C$ being close to the smooth part in \eqref{eqn:lasso_l1}. The effect of reweighting method is two-fold: (i) it is a majorization-minimization algorithm of sparse regression using  $\log$-norm as sparsity  surrogate \cite{candes2008enhancing}, hence, discovers sparse solution more effectively compared to the $\ell^1$-norm in Lasso; and (ii) the sparsity surrogate in final stages of reweighting approaches $\ell^0$-norm, by seeing $\tfrac{\mb X_{ij}^{(k+1)}}{\mb X_{ij}^{(k)}+\eps}\approx 1$ if $\mb X_{ij}^{(k)}\neq 0$ as $k\to\infty$. As a result, in the final stages, problem \eqref{eqn:lasso_l1} effectively turns into least squares, restricted to the support of $\mb X$, which produces a sparse solution with correct magnitude. \Cref{fig:reweight} (left) displays an example of reweighting scheme, showing better reconstruction result than vanilla Lasso. 

 \begin{figure}[t!]
	\centering
%
%
\begin{tikzpicture}
\newcommand{\plotwidth}{0.8in}
\newcommand{\plotheight}{0.8in}
\newcommand{\plotxone}{0in}
\newcommand{\plotxtwo}{0.9in}
\newcommand{\plotxthree}{1.8in}
\newcommand{\plotxfour}{2.7in}
\newcommand{\plotyone}{0in}

\begin{axis}[%
width=\plotwidth,
height=\plotwidth,
at={(\plotxone,\plotyone)},
scale only axis,
axis on top,
xmin=-3.025,
xmax=3.025,
ymin=-3.025,
ymax=3.025,
ytick={\empty},
xtick={\empty},
axis background/.style={fill=white},
]
\addplot [forget plot] graphics [xmin=-3.025, xmax=3.025, ymin=-3.025, ymax=3.025] {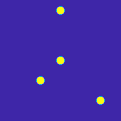};
\end{axis}

\begin{axis}[%
width=\plotwidth,
height=\plotwidth,
at={(\plotxtwo,\plotyone)},
scale only axis,
axis on top,
xmin=-3.025, 
xmax=3.025,
ymin=-3.025,
ymax=3.025,
xtick={\empty},
ytick={\empty},
axis background/.style={fill=white},
]
\addplot [forget plot] graphics [xmin=-3.025, xmax=3.025, ymin=-3.025, ymax=3.025] {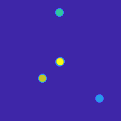};
\end{axis}

\begin{axis}[%
width=\plotwidth,
height=\plotwidth,
at={(\plotxthree,\plotyone)},
scale only axis,
axis on top,
xmin=-3.025,
xmax=3.025,
ymin=-3.025,
ymax=3.025,
xtick={\empty},
ytick={\empty},
axis background/.style={fill=white},
legend style={legend cell align=left, align=left, draw=white!15!black}
]
\addplot [forget plot] graphics [xmin=-3.025, xmax=3.025, ymin=-3.025, ymax=3.025] {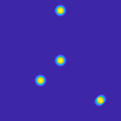};
\end{axis}

\begin{axis}[%
width=\plotwidth,
height=\plotwidth,
at={(\plotxfour,\plotyone)},
scale only axis,
axis on top,
xmin=-3.025,
xmax=3.025,
ymin=-3.025,
ymax=3.025,
xtick={\empty},
ytick={\empty},
axis background/.style={fill=white},
legend style={legend cell align=left, align=left, draw=white!15!black}
]
\addplot [forget plot] graphics [xmin=-3.025, xmax=3.025, ymin=-3.025, ymax=3.025] {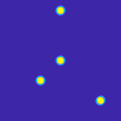};
\end{axis}

\node at (\plotxone+\plotwidth/2,\plotyone+\plotheight+0.08in) {\myfont{\footnotesize Original image}}; 

\node at (\plotxtwo+\plotwidth/2,\plotyone+\plotheight+0.08in) {\myfont{\footnotesize Lasso w/big $\lambda$}}; 

\node at (\plotxthree+\plotwidth/2,\plotyone+\plotheight+0.09in) {\myfont{\footnotesize Lasso w/small $\lambda$}}; 

\node at (\plotxfour+\plotwidth/2,\plotyone+\plotheight+0.08in) {\myfont{\footnotesize Reweight Lasso}};

\end{tikzpicture}%
	\caption{SECM image reconstruction with pure Lasso and reweighted Lasso. We apply three algorithm to reconstruct the image (left) with 6 line scans with simulated PSF in \Cref{fig:psf}. The reconstruction from Lasso with large $\lambda$ (mid left) has unbalanced magnitude due to the coherence of line scans, and from Lasso with small $\lambda$ (mid right) gives  blurry image by weakened sparsity regularizer. Reweighing Lasso can adjust the sparse regularizer in each iteration and consistently gives good result.}\label{fig:reweight}
\end{figure}
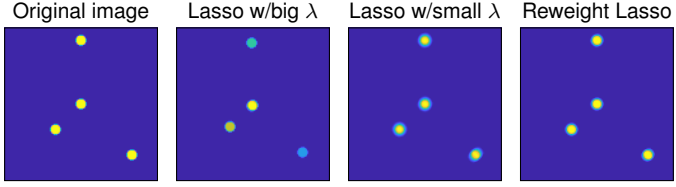

\begin{figure}[t!]
	\centering
	\input{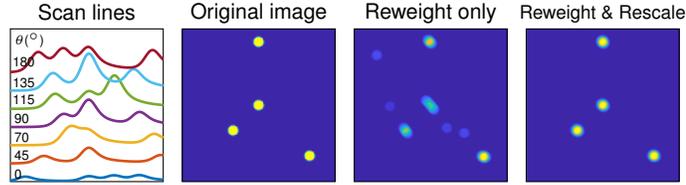}
	\caption{SECM image reconstruction with reweighed Lasso and reweighed calibrating Lasso. We simulate a line scan with uneven magnitude (left), and reconstruct the image (mid left) with two algorithm. The algorithm with reweighting only (mid right) cannot identify the correct support; where the reweighting plus calibration (right) method well approximates the image.}\label{fig:rescale}
\end{figure}

In \Cref{fig:reweight}, we display an example comparing reweighting to  the vanilla Lasso with different penalty variable in a noiseless scenario. When $\lambda$ is large, the reconstructed $\mb X$ does not recover correct relative magnitudes; when $\lambda$ is small, the effect of sparsity surrogate is weakened, resulting imprecise support recovery and offers blurry image. Using reweighting method correctly reconstruct the exact result. To show how the addressed modification in Lasso algorithm improves success rate of image reconstruction from line scans, we present a series of simulated experiments, comparing the reconstruction between the vanilla Lasso with different $\lambda$ settings and reweighting method in \Cref{fig:reweight_exps}.  Each data point consists of average of 30 experiments; in each of the experiment, the ground truth discs are generated at random with minimum separation (rejection sampling), which is then reconstructed from 8 random lines scans if disc number $<\!16$, or 16 lines scans when disc number $>\!16$. All discs are assumed to have equal magnitude. The correctness of the image reconstruction is measured by calculating the relative error between the pixel values of image, which measures the difference between normalized ground truth image and reconstructed image. The experiments show the reweighting method steadily outperforms the vanilla Lasso under various settings when measurements are incoherent.

\subsubsection{Blind calibration for incomplete  PSF information} 
Due to natural physical limitations, the incorrect estimation of PSF can be inevitable, especially in nanoscale. One remedy is to parameterize the PSF to accommodate possible variations; this can significantly improve the accuracy of reconstruction result. We assume $\mb\psi(\mb p_i)$ is a single instance of PSF with parameter $\mb p_i$, where the vector $\mb p_i$ can represent the peak value, the width of peak, and the rise/decay of PSF in \Cref{fig:psf} for the scan of angle $\theta_i$. For the reconstruction algorithm, we replace the PSF $\mb\psi$ in \eqref{eqn:lasso_l1} with the parameterized version $\mb\psi(\mb p_i)$, and optimize both the parameter $\mb p_i$ and the sparse map $\mb X$ via alternating minimization.   

\Cref{fig:rescale} exhibits a simulated example in which the PSF of line scans  has unbalanced magnitudes due to the variation of probe scanning angle. In this example, the line scan with largest overall magnitude is four times as much as the smallest, which shows the comparison of image reconstruction results from algorithm of reweighting or of reweighting plus rescaling calibration. The figures show the calibration achieves successful reconstruction while the former non-calibration method fell short on this simulated problem which has more than enough line scans are utilized to reconstruct a simplistic four disc example.

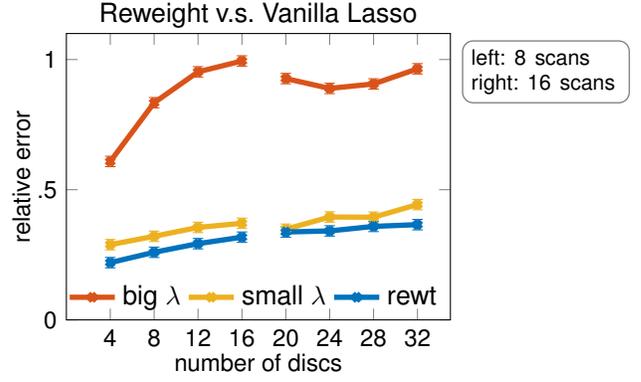
\begin{figure}[t!]
	\centering
%
%
\definecolor{mycolor1}{rgb}{0.85000,0.32500,0.09800}%
\definecolor{mycolor2}{rgb}{0.92900,0.69400,0.12500}%
\definecolor{mycolor3}{rgb}{0.00000,0.44700,0.74100}%
\newcommand{\plotwidth}{2.01in}
\newcommand{\plotheight}{1.5in}
\newcommand{\plotxone}{0in}
\newcommand{\plotxtwo}{0in}
\newcommand{\plotyone}{0in}

\begin{tikzpicture}
\sffamily{
\begin{axis}[%
width=\plotwidth,
height=\plotheight,
at={(\plotxone,\plotyone)},
scale only axis,
xmin=0,
xmax=35,
xtick={ 4,  8, 12, 16, 20, 24, 28, 32},
xticklabels={ 4,8,12,16,20,24,28,32},
x tick label style={font=\small},
ytick={0,0.5,1},
yticklabels={0,.5,1},
y tick label style={font=\small},
ymin=0,
ymax=1.1,
title={Reweight v.s. Vanilla Lasso},
title style={yshift=-0.1in},
ylabel={relative error},
xlabel={number of discs},
xlabel style={yshift=0.08in, font=\small},
ylabel style={yshift=-0.25in, font=\small}, 
axis background/.style={fill=white},
legend style={at={(1.00,0.00)}, anchor=south east, legend columns=-1, draw=none, fill=none, legend cell align=left, align=left}
]
\addplot [color=mycolor1, line width=2.0pt, mark=x, mark options={solid, mycolor1}]
 plot [error bars/.cd, y dir = both, y explicit]
 table[row sep=crcr, y error plus index=2, y error minus index=3]{%
4	0.60850504201953	0.02	0.02\\
8	0.834170875195968	0.02	0.02\\
12	0.952735068457424	0.02	0.02\\
16	0.994568801869832	0.02	0.02\\
};
\addlegendentry{$\text{big }\lambda$}

\addplot [color=mycolor2, line width=2.0pt, mark=x, mark options={solid, mycolor2}]
 plot [error bars/.cd, y dir = both, y explicit]
 table[row sep=crcr, y error plus index=2, y error minus index=3]{%
4	0.288952304386325	0.02	0.02\\
8	0.321089212132668	0.02	0.02\\
12	0.355266906200983	0.02	0.02\\
16	0.371033113950272	0.02	0.02\\
};
\addlegendentry{$\text{small }\lambda$}

\addplot [color=mycolor3, line width=2.0pt, mark=x, mark options={solid, mycolor3}]
 plot [error bars/.cd, y dir = both, y explicit]
 table[row sep=crcr, y error plus index=2, y error minus index=3]{%
4	0.219959784326033	0.02	0.02\\
8	0.259717054340004	0.02	0.02\\
12	0.292847841013869	0.02	0.02\\
16	0.317744665672588	0.02	0.02\\
};
\addlegendentry{rewt}

\addplot [color=mycolor1, line width=2.0pt, mark=x, mark options={solid, mycolor1}, forget plot]
 plot [error bars/.cd, y dir = both, y explicit]
 table[row sep=crcr, y error plus index=2, y error minus index=3]{%
20	0.92697449867536	0.02	0.02\\
24	0.889140783467807	0.02	0.02\\
28	0.906322026648134	0.02	0.02\\
32	0.964444688730697	0.02	0.02\\
};
\addplot [color=mycolor2, line width=2.0pt, mark=x, mark options={solid, mycolor2}, forget plot]
 plot [error bars/.cd, y dir = both, y explicit]
 table[row sep=crcr, y error plus index=2, y error minus index=3]{%
20	0.348269686234737	0.02	0.02\\ 
24	0.39520874386722	0.02	0.02\\
28	0.394459878094059	0.02	0.02\\
32	0.442988457861914	0.02	0.02\\
};  
\addplot [color=mycolor3, line width=2.0pt, mark=x, mark options={solid, mycolor3}, forget plot]
 plot [error bars/.cd, y dir = both, y explicit]
 table[row sep=crcr, y error plus index=2, y error minus index=3]{%
20	0.337624550858148	0.02	0.02\\
24	0.341741646642681	0.02	0.02\\
28	0.35915030383684	0.02	0.02\\
32	0.36589067261823	0.02	0.02\\
};
\end{axis} 
 
\node[rectangle, rounded corners, draw=black!50, line width=0.5pt, text width=0.78in, align=left] at (\plotwidth+0.5in,3.3) {\footnotesize left: 8 scans \\[-0.04in] right: 16 scans };

}

\end{tikzpicture}%
	\caption{Performance of reweighting method versus Lasso. We use 8 line scans when the disc number is below 16, and 16 line scans when disc number is above for reconstruction. The experiments show reweighting method outperforms vanilla Lasso with various penalty variable $\lambda$ setting  w.r.t.\ normalized (to $1$) magnitude difference between the ground truth images and reconstructed images.} 	\label{fig:reweight_exps}
\end{figure} 


\subsection{Image reconstruction algorithm}
Finally we  formally state the complete algorithm \Cref{alg:SECM-CLP-imaging}  for reconstruction of SECM image from line scans. The algorithm solves multiple iterations of 
\begin{align}\label{eqn:lasso-reweight-calibrate}
	&\min_{\mb X\geq 0,\mb p\in\mc P} \textstyle\sum_{ij}\mb \lambda_{ij}^{(k)}\mb X_{ij} \notag \\
	&\qquad +\,\textstyle\sum_{i=1}^m\tfrac12\norm{\mc S\{\mb\psi(\mb p_i)*\mc L_{\theta_i}\brac{\mb D*\mb X}\} - \mb R_i}2^2.
\end{align}
while updating the penalty variable $\mb\lambda^{(k)}$ in each iterate base on \eqref{eqn:lda_reweight}. To solve a single iterate of \eqref{eqn:lasso-reweight-calibrate}, the algorithm utilizes an accelerated alternating minimization method specifically for nonsmooth, nonconvex objectives called iPalm \cite{pock2016inertial} stated in \Cref{alg:ipalm}. Since this formulation is nonconvex and the gradients of objective \eqref{eqn:lasso-reweight-calibrate} could have large local Lipchitz constants, we adopt the backtracking method for choosing the step size of each individual gradient step. In our real data experiments, the analytic form of PSF $\wh{\mb\psi}(\mb p)$ is realized as a two-sided exponential decaying function. Define a one-side exponential-decay function as
\begin{align}
	\mc E_{c,\alpha}(t) = (c \!\cdot\! t + 1)^{-\alpha}\1_{\set{t>0}},\quad \wc{\mc E}_{c,\alpha}(t) = \mc E_{c,\alpha}(-t)
\end{align}
then $\wh{\mb \psi}(\mb p) = \wh{\mb\psi}(c_\ell,\alpha_\ell, c_r,\alpha_r,\sigma)$ is defined as
\begin{align}
	\wh{\mb\psi}(\mb p) \,=\, \brac{\, \wc{\mc E}_{c_\ell,\alpha_\ell}\, +\, \mc E_{c_r,\alpha_r} \,} * f_{\sigma}
\end{align}  
where $f_\sigma$ is zero-mean Gaussian function with deviation $\sigma$.


\begin{algorithm}[t!]\myfont{	
\caption{ $\mr{iPalm}(\mb X_{\mr{init}},p_{\mr{init}},\mb\lambda, h,\mc P)$: Accelerated iPalm for calibrating sparse regression }\label{alg:ipalm}
	\begin{algorithmic}	
		\Require Initialization $\mb X_{\mr{init}}\in\R^{n\times n}$ and $ \mb p_{\mr{init}}\in\mc P$, sparse penalty $\mb\lambda\in\R^{n\times n}$, smooth function $h$, and number of iterations $L$.
		\vspace{0.1cm}  
		\State Let ${\mb X}^{(0)} \gets  \mb X_{\mr{init}}$;\, $\mb p^{(0)}\gets p_{\mr{init}}$; $\alpha\gets 0.9$; $t_{X0},t_{p0}\gets 1$   
		\For{$\ell =1,\ldots, L $ }
		\vspace{0.10cm}
		\State // Accelerated Proximal Gradient for map $\mb X$.  
		\State $\mb Y^{(\ell)} \gets {\mb X}^{(\ell)}  + \alpha\,({\mb X}^{(\ell)} - {\mb X}^{(\ell-1)})$;\;\; $t \gets t_{X0}$;
		\Repeat{}
		\State{$t\gets t/2$};
		\State  $\mb X^{(\ell+1)}\gets\mr{Soft}^+_{t\mb \lambda}\brac{\mb Y^{(\ell)}-t\,\partial_{\mb X}h(\mb Y^{(\ell)},\mb p^{(\ell)}) } $; 
		\Until{ $h(\mb X^{(\ell+1)},\mb p^{(\ell)}) \leq   h(\mb Y^{(\ell)},\mb p^{(\ell)}) $}
		\State  $ \qquad \qquad\qquad\qquad +  \innerprod{\partial_{\mb X} h(\mb Y^{(\ell)},\mb p^{(\ell)})}{\mb X^{(\ell+1)}-\mb Y^{(\ell)}}$  
		\State $\qquad \qquad \qquad \qquad + \tfrac{1}{2t} \| \mb X^{(\ell+1)}-\mb Y^{(\ell)}\|_2^2$; 
		\State $t_{X0}\gets 4t$;
		\vspace{0.10cm}
		\State // Accelerated Proximal Gradient for parameters $p$.
		\State $\mb q^{(\ell)}\gets \mb p^{(\ell)} +\alpha\,(\mb p^{(\ell)}-\mb p^{(\ell-1)})$;\;\; $t\gets t_{p0}$;
		\Repeat{} 
		\State $t\gets t/2$;
		\State $\mb p^{(\ell+1)} \gets \mr{Proj}_{\mc P}\brac{\mb q^{(\ell)} - t\,\partial_p h(\mb X^{(\ell+1)},\mb q^{(\ell)}) }$;
		\Until{ $h(\mb X^{(\ell+1)},\mb p^{(\ell+1)}) \leq   h(\mb X^{(\ell+1)},\mb q^{(\ell)}) $}
		\State  $ \qquad \qquad\qquad\qquad +  \innerprod{\partial_p h(\mb X^{(\ell+1)},\mb q^{(\ell)})}{\mb p^{(\ell+1)}-\mb q^{(\ell)}}$  
		\State $\qquad \qquad \qquad \qquad + \tfrac{1}{2t} \| \mb p^{(\ell+1)}-\mb q^{(\ell)}\|_2^2$; 
		\State $t_{p0}\gets 4t$;
		\EndFor 
		\vspace{0.08cm}
	\Ensure $(\mb X^{(L)},\mb p^{(L)})$ as the approximated minimizers of $\min_{\mb X\geq 0,\mb p\in\mc P} \sum_{ij}\mb\lambda_{ij}\mb X_{ij} + h(\mb X,\mb p) $  
	\end{algorithmic}}
\end{algorithm}

\begin{algorithm}[t!]\myfont{
\caption{Reconstruct SECM image with line scans via reweighted iPalm.}\label{alg:SECM-CLP-imaging}

	\begin{algorithmic}
		\Require Line scans $\set{\mb R_i}_{i=1}^m$, scan angles $\set{\theta_i}_{i=1}^m$, profile $\mb D$, estimated psf $\wh{\mb\psi}$, initial guess of  parameters $\mb p_{\mr{init}}\in\mc P$ convex, and number of iterations $K$. \vspace{0.1cm} 
		\State Let $\mb X^{(0)} \gets \mb 0$, $\,\,\mb p^{(0)}\gets \mb p_{\mr{init}}$,\vspace{0.1cm}   
		\State Let $h(\mb X,\mb p) \gets \textstyle\sum_{i=1}^m\tfrac12\|\wh{\mb\psi}*\mc L_{\theta_i}[\mb p]\brac{\mb D*\mb X} - \mb R_i\|_2^2$;\vspace{0.1cm}  
		\For{$k = 1,\ldots, K$}
		\If{$k=1$}
			$\mb\lambda \gets C\max_{ij}\set{\mc L_\Theta^*\big[\wh{\mb\psi}*\mb R\big]_{ij}}\cdot \1  $;
		\Else
			\State $\forall\,i,j\in[n]$, $\mb \lambda^{(k)}_{ij}\gets Ch(\mb X^{(k)},\mb p^{(k)})/(\mb X^{(k-1)}_{ij} + \eps)$;
		\EndIf 
		\vspace{0.1cm}
		\State $(\mb X^{(k+1)},\mb p^{(k+1)}) \gets \mr{iPalm}(\mb X^{(k)},\mb p^{(k)},\mb\lambda,\, h,\mc P)$;
		\EndFor
		\Ensure Reconstructed image $\mb Y\gets \mb D *\mb X^{(K)} $
	\end{algorithmic}}
\end{algorithm}

\section{Real data experiments}\label{sec:real-data}

\begin{figure}[t!]
\centering
\input{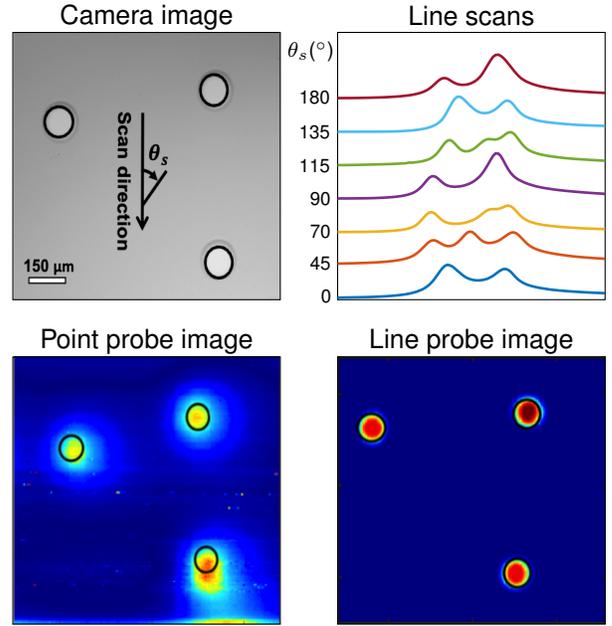}
\caption{Real signal experiments on three platinum discs \cite{dorfi2019instrument}. We show the reconstruction result of a three disc sample (up-left), which is scanned with line probe in seven different directions (up-right). The arrow in optical image represents the line probe sweeping direction, while as $\theta_s$ stands for clockwise rotation of the sample. The black circle indicates the correct disc location in each images. Compare to the point probe, in which the shifts of disc location are resulted from the skew of PSF (down-left), our line scan reconstruction accurately recovers the exact location (down-right). For both of the reconstructed images, the resolution is $10\mr{\mu m}$ per pixel.} \label{fig:3dots} 
\end{figure}

We present two sets of experiments to demonstrate an end-to-end result of line probe SECM. 

\Cref{fig:3dots} displays the comparison of the line probe/point probe scan on a simplistic three disc samples ($75\mr{\mu m}$ in radius, platinum). In these experiments, the point probe tip diameter and the line probe edge thickness are equivalent ($\approx\! 20 \mr{\mu m}$), and the probe moving speed ($100\mr{ms}$), the sampling rate ($10\mr{\mu m}$), and the probe end material (platinum) are identical as well. Four images are shown here, including the optical closeup image for the three  discs, the line scans, and the reconstruction image of either point probe or the line probe. In the optical image, the arrow (scan direction) represents the line probe sweeping direction when $\theta_s = 0^\circ$, which generates the  $0^\circ$ line scan. The three discs sample is then rotated by $\theta_s$ ($45^\circ$ in this case) clockwise, proceeds with another sweep of line probe, produces the $45^\circ$ line scans. This routine continues  until all seven scans are carried out. 

In the reconstructed images, the black circles indicate the ground truth size and location of the platinum discs derived from the optimal image. The reconstruction algorithm \Cref{alg:SECM-CLP-imaging} is setup with 6 reweighting iterates, where each iterates runs 50 iterates of iPalm. We can see the reconstructed result from the point probe exhibits distortion in the image due to the skewness of probe PSF along its proceeding direction during raster scans; while the image of line scan reconstruction presents three circular features with its size and locations are agreeing with the ground truth, since the skewness of PSF has been successfully corrected by the reconstruction algorithm.

In \Cref{fig:8_10dots}, we reconstruct images of samples consisting of platinum discs arranged in more complicated configurations. Two  experiments are presented here, which are the samples consisting of 8 or 10 discs, while the disc diameter/image resolution/probe dimension/sampling rate are all identical to the three discs case in \Cref{fig:3dots}. The reconstruction algorithm are also set up similarly, with reweighting (ipalm) procedures with 6(50) iterates, generating the images of interest of much larger dimension. Notice that here we use 7(9) line scans on 8(10) disc sample respectively, and demonstrate both of the resulting reconstructed image and the location map, in which the location map is a binary image defined by $\1_{\set{\mb X_{ij}\geq 0.5\norm{\mb X}\infty}}$ at $(i,j)$-th entry.
 
 We can see for these more complicating images, our algorithm are still able to reconstruct the image of platinum discs with correct location and shape. The corresponding location maps are  approximately recovered, with most of the  discs locations are represented by a single one-sparse vector, and some other locations are represented by a two-sparse vector due to the inevitable discretization error.   

Our code for the reconstruction of SECM image from scans from line probe can be found via the following link:

\vspace{0.03in}
\begin{center}
	\noindent\texttt{\url{https://github.com/clpsecm/clpsecm_imaging}}
\end{center}

\begin{figure*}[t!]
\centering
\input{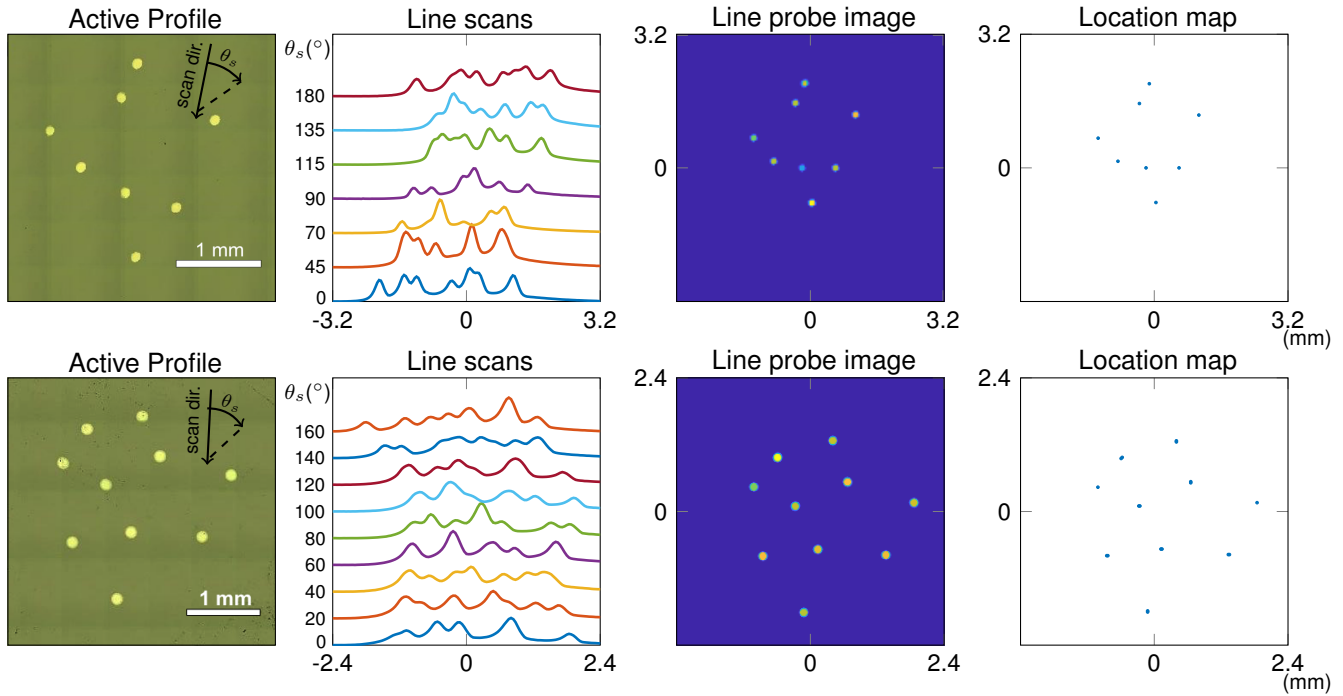}
\caption{Real signal experiments of 8, 10 platinum discs. Showing the optimal image of the 8 discs (up) and 10 discs (down) sample, and their corresponding line scans, reconstructed image and reconstructed disc location map. In optical image, the arrows represent the line probe sweeping direction, while as $\theta_s$ stands for clockwise rotation of the sample. In both examples, our algorithm is able to successfully obtain these images of the discs, with most of the disc locations can be approximately represented by an one-sparse vector. Here, the image resolution is $20 \mr{\mu m}$ per pixel. }  \label{fig:8_10dots}
\end{figure*}

\section{Summary \& Discussion}
This paper presents the development of a novel scanning probe microscope technique based on line measurements. The microscope obtains line integrals in each measurement, such measurements are non-local, hence more efficient then conventional raster scans for microscopic image with localized sparse structure. This paper shows the improved efficiency of line probe via rudimentary analysis and experiments; and proposes a simple modification in conventional CS algorithm for image reconstruction, with its effect on both the simulated and the actual datasets. Due to the strong relation between  computational tomography and line scans, we also view our work can potentially being applied to ares of CT or other similar imaging modalities involving the use of projection measurements.

We envision multiple possibilities for future work. First, the current studied microscopic images are circumscribed in sparse convolutional model; while it has an immediate access to applications such as lattice structure imaging in material science, we aim to expand the potential application of line scans to more general imaging problems. Furthermore, unlike many other imaging modalities, in SPM the design of probe topography (i.e. the sampling pattern) is not limited to a straight line,  therefore it is possible adopt various different probe design to achieve CS-like sample reduction. Lastly, in this paper we have shown via simple reasoning and experiments to exhibit the relationship between the complexity of image and the required number of line scan measurements to achieve exact reconstruction. We consider rigorously demonstrating the relationship can also be an interesting direction in CS, especially since the line scans are not the CS optimal measurement model.


\nocite{*}


\end{document}